\documentclass[a4paper]{article}

\usepackage[margin=2cm]{geometry}	% to adjust page margins
\usepackage{tikz}		% drawing quantum circuits
\usepackage{amsmath}
\usepackage{amssymb}
\usepackage{amsthm}
\usepackage{comment}
\usepackage{breqn}
\usepackage{graphicx}
\usepackage{subcaption}

\title{Grover's Algorithm and Many-Valued Quantum Logic}
\author{Samuel Hunt and Maximilien Gadouleau}
\date{\today}

\usetikzlibrary{backgrounds,fit,decorations.pathreplacing,calc}
\newcommand{\ket}[1]{\ensuremath{\left|#1\right\rangle}}
\newcommand{\bra}[1]{\ensuremath{\left\langle#1\right|}}
\newcommand{\braket}[2]{\ensuremath{\left\langle#1 | #2\right\rangle}}
\newcommand{\bramatket}[3]{\ensuremath{\left\langle#1 | #2 | #3\right\rangle}}

\newtheorem*{theorem*}{Theorem}

\newtheorem{theorem}{Theorem}[section]

\newtheorem{lemma}[theorem]{Lemma}

\begin{document}

  \maketitle

\begin{abstract}
\noindent As the engineering endeavour to realise quantum computers progresses, we consider that such machines need not rely on binary as their de facto unit of information. We investigate Grover's algorithm under a generalised quantum circuit model, in which the information and transformations can be expressed in any arity, and analyse the structural and behavioural properties while preserving the semantics; namely, searching for the unique preimage to an output a function. We conclude by demonstrating that the generalised procedure retains $O(\sqrt{N})$ time complexity.
\end{abstract}

  \section{Introduction} \label{section_introduction}

Lov Grover devised his eponymous algorithm in 1996 as a quantum search technique \cite{Grover}, permitting a quantum computer to derive the solution to a problem of size $N$ in $O(\sqrt{N})$ time in the worst-case. This is faster than any classical worst-case, which must sequentially evaluate all $N$ elements, and is thus an example of \textit{quantum supremacy}\footnote{This refers to a quantum process which solves a problem asymptotically faster than any classical counterpart.}; in fact, Grover's algorithm was the first to surely exhibit this property in an applicable context\footnote{The first to display quantum supremacy was the Deutsch-Jozsa algorithm, published in 1992 \cite{DJ}, of little practical use. Shor's algorithm was published in 1995 \cite{Shor}, showing that \texttt{Integer Factorisation} belongs to \texttt{BQP}. However, while is suspected this problem doesn't belong to either \texttt{P} or \texttt{BPP}, this is not yet known.}. Formulated precisely, the algorithm takes a function $f$ over a codomain $S$ and some output of significance $y \in f(S)=\left\{f(s):s\in S\right\}$, and finds the unique preimage $x$ such that $f(x)=y$ in at most $O(\sqrt{S})$ calls to the function $f$. The algorithm can be modified to work with multiple targets or conduct what Grover and Radhakrishnan call partial search, but we don't consider such variations.

We investigate Grover's algorithm within a generalisation of the \textit{quantum circuit model}, in which the arity of information and consequent transformations are not limited to binary. The model considers a unit of information with $d$-many distinct states and quantum transformations between these levels. The bit has remained the de facto unit of classical digital computing due to the electronic simplicity of the transistor, with some modern transistors residing just above the quantum scale. Essentially a switch, the transistor either permits or inhibits current flow, naturally endowing binary\footnote{The bit existed before electronics as punched-cards, invented in 1732; Shannon coined it the \textit{bit} in his seminal 1948 paper.}. On the other hand, no de facto quantum system has presently been engineered for computation, and while there is an endeavour to realise the qubit - the quantum counterpart to the bit -  quantum systems are not restricted in their arity by nature. In fact, as with quantum computation in general, the only limitation is our ability to control these systems. Physically, quantum states are just nominated energy levels within quanta, and with enough control over the system we can nominate and navigate through an arbitrary number of them, encapsulating any arity.

While current research is focused primarily on the physical realisation of the qudit [4-8], we feel that it is important to simultaneously develop our knowledge of how this technology can be utilised. Previous consideration has been given to algorithms within this generalisation [9-12], with results upon which we hope to build. The work in \cite{MultiAlg}, in particular, discusses Grover's algorithm, but lacks a thorough analysis of structure and behaviour, with important and informative details subsequently overlooked, such as the generalised operators' invariance under a $d$-dimensional subspace and the expected number of runs as $d$ increases. We present a rigorous and fully investigative first analysis of a useful quantum algorithm within this framework.

The paper is structured as follows. Section \ref{section_model} briefly highlights some relevant features of the quantum circuit model and formally introduces the many-valued generalisation. Section \ref{section_Grover} details a thorough treatment of Grover's algorithm in binary, proving the construction of the circuit and derivation of the time-complexity. Section \ref{section_d-ary} presents the work novel to this paper. We first explain our approach and its mathematical intuition. We provide the circuit and detail its linear algebraic structure. We then proceed to analyse its behavioural properties, such as an invariant subspace. We provide two final analyses for proving the desired result, that the generalised circuit performs the same function as its binary counterpart and retains sublinear time complexity: the first, mirroring precisely the binary method, in the ternary case; the second, following a different, `synthetic' approach in the general $d$-ary case. Finally, Section \ref{section_conclusion} summarises the results of this paper and considers where subsequent work in this field might be directed.

  \section{The Many-Valued Quantum Circuit Model} \label{section_model}
This section briefly explains the generalisation of the quantum circuit model presented in \cite{Arity}. Before proceeding, we note that this is not the only model considered for implementable quantum computation\footnote{Other notable models include the adiabatic quantum computer and topological quantum computer.}, but its prominence is due to its cooperative nature and familiarity to classical digital computing. Moreover, the quantum Turing machine has been precisely defined, constituting a universal model of quantum computation \cite{QTM}, and it can be shown that certain - although infinite - gate-sets do exhibit universality within quantum computing \cite{Circ1}, and that more restrictive and realistic gate sets permit the model to approximate universality with arbitrary precision \cite{Circ2}. Similar results exist for qudits; for example, the considered generalisations of the Pauli $X$ and $Z$ operators are universal for single-qudit operations \cite{UnivQudit}. Crucially, working within this framework permits us to utilise the full breadth of quantum computation and thus the exploration of this model and its extensions is important. Finally, this generalisation does not seem to have a set-in-stone name and so we refer to it as \textit{many-valued quantum logic}, in-line with \cite{Treatise}.

We assume the reader is familiar with quantum computing and the standard quantum circuit model, but highlight particular features of importance to this work; for a complete introduction, we recommend \cite{Gentle}, otherwise there is much material on the subject. Firstly, the normalisation of quantum states entails that they lie as points on the unit hypersphere and thus transformations between states detail rotations, which we can capture with linear algebra; these transformations are equivalent to unitary linear operators. Consequently, any single-qubit operator can be detailed by a $2\times2$ unitary, complex matrix. Following are some qubit operators of particular note, due to both their presence in the circuit for Grover's algorithm and fundamental status in quantum computing: firstly, the quantum $\text{NOT}$ gate; secondly, the quantum phase-shift gate; finally, the Hadamard gate i.e. binary quantum Fourier transform.
  \begin{align*}
X = \left(\begin{matrix} 0 & 1 \\ 1 & 0 \end{matrix}\right), Z = \left(\begin{matrix} 1 & 0 \\ 0 & -1 \end{matrix}\right), H = \frac{1}{\sqrt{2}}\left(\begin{matrix} 1 & 1 \\ 1 & -1 \end{matrix}\right)
  \end{align*}

Secondly, quantum entanglement  leads us to the construction of multi-qubit gates that cannot be expressed per single-qubit gates; such are our tools to instill this phenomenon. The most common multi-qubit gate is the controlled NOT gate, enacted in its most simple form across two qubits. In this case, its matrix is as follows.
  \begin{align*}
C_{\text{NOT}} = \left(\begin{matrix} I & 0 \\ 0 & X \end{matrix}\right)
  \end{align*}
We can think of the two qubits as being, individually, the \textit{control} and the \textit{target}. When applied, the target qubit is subject to the NOT gate only when the control is in the $1$ state. Let us briefly demonstrate that $C_{\text{NOT}}$ cannot be constructed per single-qubit gates. Recall that the matrix corresponding to the application of a gate-array is given by the tensor product of the composite gates. There are then no two matrices $U_1$ and $U_2$ such that $U_1 \otimes U_2 = C_{\text{NOT}}$, as we would require that both $U_1U_2=I$ and $U_1U_2=X$.

We combine meaningful sequences of these gates to achieve a particular task, with the sequence constituting a \textit{quantum circuit}; thus, we have the quantum circuit model realisation of a \textit{quantum algorithm}. The subject of this paper is Grover's quantum algorithm, which contains two important sub-circuits, the \textit{diffusion} and \textit{black-box} operators, with further auxiliary quantum components to permit easy extraction of the solution. We will address the algorithm in detail in Section $3$.

The generalisation defines discrete quantum states of any arity and generalises the gates in line with this more expressive unit of information. We first define this generalised unit, the \textit{qudit}. The extension is simple: rather than nominating two distinct quantum states between which we transverse, we nominate $d$-many and regard the system as inhabiting a $d$-dimensional vector space. The qudit is then specified as follows. Note that, we still mandate normalisation and prefer orthogonal base states, for the same reasons as in the binary case: direct extraction of probabilities and ease of measurement.
  \begin{align*}
\ket{\psi_d} = \sum_{i=0}^{d-1} a_i \ket{i} \ :\ \sum |a_i|^2 = 1
  \end{align*}

We now define the generalised operators. Abstractly, the  binary operators navigate qubits between different states and phases. We typically restrict binary operators to enact two phase-shifts, $1$ and $-1$, with these values nominated because they are maximally far apart on the unit circle, which can denote phase. Following this intuition in generalising to $d$ distinct phases for a $d$-ary system, we desire that these phases all be equally far apart. We achieve this per the $d$th roots of unity, complex numbers $z$ that satisfy $z^d=1$. We can define the solutions per trigonometry in which the primitive $d$th root is given as follows. All solutions can be obtained through raising the primitive to a power $k\in[0,d-1]$.
  \begin{align*}
\omega_d=\exp\left(\tfrac{2\pi i}{d}\right)
  \end{align*}
We drop the subscript when the dimension is implicit. Each set of $d$th roots form a cyclic group, distributed uniformly across the unit circle, satisfying precisely our qudit phase-shift requirements. To conclude, when considering discrete phase-shifts in a $d$-ary quantum logic, we use the $d$th roots of unity.

Consider the three $d$-ary gates that generalise our considered single-qubit binary gates: $X_d$, $Z_d$, and $F_d$. The NOT gate $X$ is generalised to the increment gate; the phase-shift gate $Z$ is generalised to apply one of $d$ phases depending on the qudit value; the binary quantum Fourier transform - Hadamard gate - $H$ is considered in its general case.
  \begin{align*}
X_d = \left(\begin{matrix} 0 & 0 & \ldots & 0 & 1 \\ 1 & 0 & \ldots & 0 & 0 \\ 0 & 1 & \ldots & 0 & 0 \\ \vdots & \vdots & \ddots & \vdots & \vdots \\ 0 & 0 & \ldots & 1 & 0\end{matrix}\right),\ Z_d = \left(\begin{matrix} 1 & 0 & 0 & \ldots & 0 \\ 0 & \omega & 0 & \ldots & 0 \\ 0 & 0 & \omega^2 & \ldots & 0 \\ \vdots & \vdots & \vdots & \ddots & \vdots \\ 0 & 0 & 0 & \ldots & \omega^{d-1}\end{matrix}\right), F_d = \left(\begin{matrix} 1 & 1 & 1 & \ldots & 1 \\ 1 & \omega & \omega^2 & \ldots & \omega^{d-1}  \\ 1 & \omega^2 & \omega^4 & \ldots & \omega^{2(d-1)}  \\ \vdots & \vdots & \vdots & \ddots & \vdots \\ 1 & \omega^{d-1} & \omega^{2(d-1)} & \ldots & \omega^{(d-1)(d-1)} \end{matrix}\right)
  \end{align*} 
As in the binary case, we must also define multi-qudit gates to utilise the phenomenon of quantum entanglement. For example, we could generalise the controlled NOT gate to the SUM gate: $\text{SUM}\ket{m}\ket{n}=\ket{m}\ket{m+n \mod d}$. As with the extended controlled NOT, we could then extend the SUM to operate over $n$-many qudits by adding to the target the product of the controls.  In Section $4$, we define an extension of the controlled NOT gate necessary to complete our generalisation of Grover's algorithm. Note, the dimension subscript of gates or operators is often dropped when it can be obviously inferred. We will typically highlight the dimension to which we refer, or that we are working in the general-case, at the start of a section and drop the subscript unless it is necessary throughout the remainder.

  \section{Grover's Algorithm in Binary} \label{section_Grover}
As previously stated, Grover's algorithm determines the unique preimage to a given output of a function over a codomain of size $N$, using $O(\sqrt{N})$ calls of the function. This is equivalent to finding the element of a space satisfying a given property faster than the worst-case, trial-and-error solution necessitated in classical computing. We can therefore think of Grover's algorithm as being able to brute force the solution to a problem faster than a classical computer. Crucially, the speedup is only $O(\sqrt{N})$ and can't be used to solve \texttt{NP} problems in polynomial time. The structure of importance is the codomain of the function, and in solving the problem we embed and manipulate this structure in the quantum world.

Thus, let us consider how Grover went about quantum brute-force over a search-space of size $N=2^n$. The algorithm begins by invoking a quantum superposition per applying the Hadamard gate across an $n$-qubit register set to state $\ket{0}$; the operator superposes the register into a state detailing each standard-basis vector with equal magnitude. Any computation rendered on the superposition is then rendered on each individual basis vector, i.e. on all $2^n$ elements. We then define a conditional phase-shift operator, enacted only when the argument matches a nominated value; or, equivalently, when the image of the argument under a function equals a nominated value. In Grover's algorithm, the argument is subjected to a phase shift of $-1$  when its state, under the function of note, matches the image that we are trying to reverse. In the superposition, this operator acts individually on each basis state, solely phase-shifting the standard-basis element  matching our image. We denote this the \textit{black-box operator}.

Now, it is important not to the fall into the trap of wondering why we can't immediately extract the preimage, given that we can instantly single it out. We know that the extraction of information from a quantum system is done through measurement, which collapses the system into a state orthogonal to the measurement axes. Moreover, the likelihood of an outcome is equal to its projected amplitude in the corresponding plane. Thus, to vary the outcome probabilities of a system under measurement, we must vary the relative amplitudes of the system apropos the measurement axes. Inverting the sign of a quantum state has no effect on its outcome under measurement, as its amplitude relative to any axis remains invariant. Subsequently, having singled out our desired preimage, we must employ some mechanism through which we increase its relative amplitude. Furthermore, any measurement will erase all information about the prior, non-reconstructible quantum system. Thus, when we make this measurement, we want to be sure that it will yield our desired result. These important details are summarised by Scott Aaronson as the number one thing to take away from his blog: ``Quantum computers would not solve hard search problems instantaneously by simply trying all the possible solutions at once'' \cite{Blog}.

Grover managed to derive an operator that increases the amplitude of the phase-segregated preimage while diminishing all other elements. Taking inspiration from the heat diffusion equation, detailing that heat flows from places where it is higher to places where it is lower, gradually averaging, he defined a transformation that reflected the amplitude of each basis-vector in a superposed state about the mean, denoted the \textit{diffusion operator}. After application of the phase-shift, the sequestered vector possesses negative amplitude, while all other elements have equal positive amplitude. Consequently, the mean amplitude is valued at just below the positive elements, and reflection about this diminishes the positive vectors while inverting and amplifying the preimage. However, this amplification is periodic; past a certain point, the reflection begins to diminish the preimage and amplify the remaining elements, until the cycle begins anew. We want to make a measurement when we will be sure that the outcome will be correct with high probability.

Formally, the algorithm initialises the system to the entangled state 
\[
\ket{s} = 1/\sqrt{N} \sum_{x=0}^{N-1} \ket{x} = H^{\otimes n} \ket{0}.
\]
The Grover operators are then applied $r$-many times, with each application denoted to be a single \textit{Grover Iteration}. Finally, a measurement is performed, collapsing the system onto the preimage $\ket{\tau}$ with probability approaching $1$ for $N \gg 1$. To guarantee this, we derive a value for $r$ that maximises the system's amplitude along the $\tau$ plane. In the remainder of this section, we prove a circuit construction of the diffusion operator and proceed to derive a value on $r$.

  \subsection{Circuit Construction} 
Consider that we can express the aforementioned operators as projections, denoting the black-box operator $U_\tau$ and the diffusion operator $U_s$.
  \begin{align*}
U_\tau = I - 2\ket{\tau}\bra{\tau}\ ,\  U_s = 2\ket{s}\bra{s} - I
  \end{align*}
The former dictates that the basis element corresponding to the primage $\tau$ is locally phase-shifted through $-1$, while the remaining basis elements are left unchanged; this can also be thought of as a reflection around $\ket{\tau}$. Similarly, the latter details a reflection about $\ket{s}$. To show this, consider that the state $\ket{\psi}=a\ket{s} + b\ket{s^\perp}$, in which $\perp$ denotes orthogonality, yields $U_s\ket{\psi}=a\ket{s} - b\ket{s^\perp}$ under the operator. We will see that in combining these operators together, reflecting the state around both $\tau$ and $s$, we edge closer to $\tau$. 

In our circuit construction, we present a circuit specifying the diffusion operator, but leave the former black-box operator as such; although its circuit construction over a considered function is of course known \cite{Oracle}. Firstly, note that we can re-express the diffusion operator per conjugation with the Hadamard gate.
  \begin{align*}
U_s = H^{\otimes n} \cdot \left( 2\ket{0}\bra{0} - I \right) \cdot H^{\otimes n}
  \end{align*}
This conjugation is easily constructed, per local application of the Hadamard across the register before and after the internal details. Thus, we are left to construct a circuit specifying $2\ket{0}\bra{0} - I$, which is accomplished through the circuit presented in Fig. \ref{BinNoConjH}.

 % --- --- ---  --- --- --- --- --- --- ---  --- --- --- --- --- --- --- ---  --- --- --- %
  % --- --- --- Binary Quantum Circuit without Conjugacy by H --- --- --- %
 % --- --- ---  --- --- --- --- --- --- ---  --- --- --- --- --- --- --- ---  --- --- --- %
  \begin{figure}[!h]\label{BDC2}
      \begin{center}
      \begin{tikzpicture}[thick]
          % `operator' will only be used by Hadamard (H) gates here.
          % `phase' is used for controlled phase gates (dots).
          % `surround' is used for the background box.
          \tikzstyle{operator} = [draw,fill=white,minimum size=1.5em] 
          \tikzstyle{phase} = [draw,fill,shape=circle,minimum size=5pt,inner sep=0pt]
          \tikzstyle{cnot} = [draw,shape=circle,minimum size=8pt,inner sep=0pt]
          \tikzstyle{surround} = [fill=white!10,thick,draw=black,rounded corners=0mm]
          \matrix[row sep=0.4cm, column sep=0.5cm] (circuit) {
              % --- --- --- First Row --- --- --- %
              \node[operator] (X12) {$X$}; &
              &
              \node[phase] (P14) {}; &
              &
              \node[operator] (X16) {$X$}; & \\ %{$X^{-1}$}; &
              %\coordinate (end1); \\

              % --- --- --- Second Row --- --- --- %
              \node[operator] (X22) {$X$}; &
              &
              \node[phase] (P24) {}; &
              &
              \node[operator] (X26) {$X$}; & \\
              %\coordinate (end2); \\

              % --- --- --- Third Row --- --- --- %
              \node[operator] (X32) {$X$}; &
              \node[operator] (H33) {$H$}; &
              \node[cnot] (P34) {}; &
              \node[operator] (H35) {$H$}; &
              \node[operator] (X36) {$X$}; & \\
              %\coordinate (end3); \\
          };
         \begin{pgfonlayer}{background}
           % Draw background box.
           %\node[surround] (background) [fit = (H11) (H31) (bracket)] {};
           % Draw lines.
           \draw[thick] (X12) -- (X16)  (X22) -- (X26) (X32) -- (X36) (P14) -- (P24) (P24) -- (P34);
         \end{pgfonlayer}
      \end{tikzpicture}
      \end{center}
      \caption{Quantum circuit realising Grover's diffusion operator $2\ket{s}\bra{s}-I$, without conjugation by the Hadamard; i.e. $2\ket{0}\bra{0} - I$. The circuit is given for 3 qubits but is generalised to $n$-many by `duplicating the top row upwards.'}
      \label{BinNoConjH}
  \end{figure}
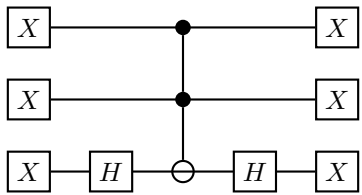

\begin{theorem}[Grover Diffusion Circuit]
  The circuit in Fig. \ref{BinNoConjH} realises $2\ket{0^n}\bra{0^n} - I_{2^n}$ with a negligible global phase shift of $-1$. 
\end{theorem}

Note, $0^n$ here denotes an $n$-qubit state in which all qubits are set to $0$ and thus we derive an operator of size $2^n \times 2^n$. However, we must first extend our definition of the controlled NOT gate to span $n$-many qubits, denoted $C_X^n$, in which the first $n-1$ are all controls and the final qubit is the target. The operator is defined such that the target is subjected to the $X$ gate if and only if all of the control bits are \ket{1}. As a base case, $C_X^1=X$. With this defined, we verify the circuit, encapsulated in the following linear algebra.
\begin{align*}
    2\ket{0^{n}}\bra{0^n} - I_{2^n} = (-1) \cdot X^{\otimes n} \cdot (I^{\otimes n-1} \otimes H) \cdot C_X^n \cdot (I^{\otimes n-1} \otimes H)  \cdot   X^{\otimes n}
\end{align*}

\begin{proof}
We prove this in a recursive manner, beginning with the base case on $n=1$.
  \begin{align*}
  \begin{split}
   X \cdot H \cdot X \cdot H \cdot X &=
      \left( \begin{matrix} 0 & 1 \\ 1 & 0 \\ \end{matrix} \right) \cdot
      \left( \begin{matrix} 1 & 1 \\ 1 & -1 \\ \end{matrix} \right) \cdot
      \left( \begin{matrix} 0 & 1 \\ 1 & 0 \\ \end{matrix} \right) \cdot
      \left( \begin{matrix} 1 & 1 \\ 1 & -1 \\ \end{matrix} \right) \cdot
      \left( \begin{matrix} 0 & 1 \\ 1 & 0 \\ \end{matrix} \right) \\
    &= \left( \begin{matrix} -1 & 0 \\ 0 & 1 \\ \end{matrix} \right) = (-1) \cdot \left(2\ket{0}\bra{0} - I\right) = (-1) \cdot \left(2\left(\begin{matrix} 1 & 0 \\ 0 & 0 \end{matrix}\right) - \left(\begin{matrix} 1 & 0 \\ 0 & 1 \end{matrix}\right)\right)
  \end{split}
  \end{align*} 

Let us now recursively define the operators under the tensor product: $X^{\otimes n}$, $I^{\otimes n-1} \otimes H$, and $C_X^n$, recalling that concurrent, local application of gates across a register is encapsulated as such.
  \begin{align*}
    X^{\otimes n} &= \left(\begin{matrix}
        0 &  X^{\otimes n-1} \\
         X^{\otimes n-1} & 0 \\
      \end{matrix}\right), \;
    I^{\otimes n-1} \otimes H = \left(\begin{matrix}
         I^{\otimes n-2} \otimes H & 0 \\
         0 & I^{\otimes n-2} \otimes H \\
      \end{matrix}\right), \;
    C_X^n = \left(\begin{matrix}
         I & 0 \\
         0 & C_X^{n-1} \\
      \end{matrix}\right)
  \end{align*}
Substituting these into our circuit encapsulation, we obtain the following.
  \begin{multline*}
    X^{\otimes n} (I^{\otimes n-1} \otimes H) C_X^n (I^{\otimes n-1} \otimes H) X^{\otimes n} \\ =
      \left(
      \begin{matrix}
        X^{\otimes n-1} (I^{\otimes n-2} \otimes H) C_X^{n-1} (I^{\otimes n-2} \otimes H) X^{\otimes n-1} & 0 \\
        0 & X^{\otimes n-1} (I^{\otimes n-2} \otimes H) (I^{\otimes n-2} \otimes H) X^{\otimes n-1} \\
      \end{matrix}
      \right) \label{bin_circuit_rec}
  \end{multline*}
The above details a bottom-up structure, from which it is clear that the circuit specifies $(-1)\cdot \left(2\ket{0^n}\bra{0^n}-I\right)$, as the base case guarantees the top-left entry recurses down to $-1$ and the remaining entries are the identity matrix, $I = X^{\otimes n-1} \cdot (I^{\otimes n-2} \otimes H) \cdot (I^{\otimes n-2} \otimes H) \cdot X^{\otimes n-1}$.
\end{proof}

To reiterate, the circuit in Fig. 1 has a global phase shift of $-1$ when compared to the initially considered projection $2\ket{0}\bra{0} - I$, but this is of no consequence to the behaviour of the operator. Thus, we conclude our circuit construction for the diffusion operator by factoring in the register conjugation through the Hadamard, as detailed in Fig. 2.

 % --- --- ---  --- --- --- --- --- --- ---  --- --- --- --- --- --- --- ---  --- --- --- %
  % --- --- --- Quantum Circuit for Grover's Algorithm in Binary --- --- --- %
 % --- --- ---  --- --- --- --- --- --- ---  --- --- --- --- --- --- --- ---  --- --- --- %
  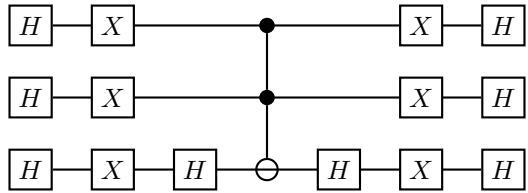
\begin{figure}[!ht]\label{BDC1}
      \begin{center}
      \begin{tikzpicture}[thick]
          % `operator' will only be used by Hadamard (H) gates here.
          % `phase' is used for controlled phase gates (dots).
          % `surround' is used for the background box.
          \tikzstyle{operator} = [draw,fill=white,minimum size=1.5em] 
          \tikzstyle{phase} = [draw,fill,shape=circle,minimum size=5pt,inner sep=0pt]
          \tikzstyle{cnot} = [draw,shape=circle,minimum size=8pt,inner sep=0pt]
          \tikzstyle{surround} = [fill=white!10,thick,draw=black,rounded corners=0mm]
          \matrix[row sep=0.4cm, column sep=0.5cm] (circuit) {
              % --- --- --- First Row --- --- --- %
              \node[operator] (H11) {$H$}; &
              \node[operator] (X12) {$X$}; &
              &
              \node[phase] (P14) {}; &
              &
              \node[operator] (X16) {$X$}; & %{$X^{-1}$}; &
              \node[operator] (H17) {$H$}; & \\
              %\coordinate (end1); \\

              % --- --- --- Second Row --- --- --- %
              \node[operator] (H21) {$H$}; &
              \node[operator] (X22) {$X$}; &
              &
              \node[phase] (P24) {}; &
              &
              \node[operator] (X26) {$X$}; &
              \node[operator] (H27) {$H$}; & \\
              %\coordinate (end2); \\

              % --- --- --- Third Row --- --- --- %
              \node[operator] (H31) {$H$}; &
              \node[operator] (X32) {$X$}; &
              \node[operator] (H33) {$H$}; &
              \node[cnot] (P34) {}; &
              \node[operator] (H35) {$H$}; &
              \node[operator] (X36) {$X$}; &
              \node[operator] (H37) {$H$}; & \\
              %\coordinate (end3); \\
          };
         \begin{pgfonlayer}{background}
           % Draw background box.
           %\node[surround] (background) [fit = (H11) (H31) (bracket)] {};
           % Draw lines.
           \draw[thick] (H11) -- (H17)  (H21) -- (H27) (H31) -- (H37) (P14) -- (P24) (P24) -- (P34);
         \end{pgfonlayer}
      \end{tikzpicture}
      \end{center}
      \caption{Quantum circuit detailing the binary Grover diffusion operator, $2\ket{s}\bra{s} - I$. The circuit is again presented for 3 qubits but can be extended as detailed in Fig. 1.}
      \label{Fig. 1}
  \end{figure}

  \subsection{Behavioural Analysis}
With the diffusion and black-box operators specified, we now consider the full computation of Grover's algorithm. As previously stated, alongside some auxiliary quantum components, the algorithm is achieved through iteratively applying the black-box then diffusion operators. Through their repeated application, the amplitude of the $\tau$ basis element periodically magnifies and diminishes. We want to perform a measurement after $r$-many iterations such that $r$ maximises the probability of the measurement collapsing the system into $\tau$. We will see that we can express $r$ as a function on $N$, and that this function details an algorithm asymptotically faster than anything achievable on a classical computer. The entire computation can be seen in Fig. 3.

\begin{figure}[h!]
  \begin{center}
  \includegraphics[scale=0.75]{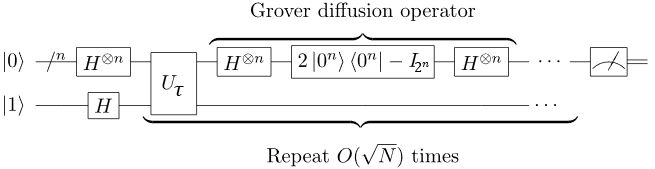}
  \end{center}
  \caption{The full circuit for Grover's algorithm in binary.}
  \label{fullcircuit}
\end{figure}

Before considering how we derive such a value, let us note that due to the system being initialised to state $\ket{s}$ and nature of the projection operators, we remain in a linear combination of $\ket{s}$ and $\ket{\tau}$ throughout computation. Formally, defining the Grover Subspace $G=\langle \ket{s},\ket{\tau} \rangle$, $G$ forms an invariant subspace under the operators $U_s$ and $U_\tau$.
 \begin{align*}
U_s \ket{s} = \ket{s},\ U_s \ket{\tau} = (2/\sqrt{N})\ket{s} - \ket{\tau} \\
U_\tau \ket{s} = \ket{s} - (2/\sqrt{N})\ket{\tau},\ U_\tau \ket{\tau} = -\ket{\tau}
  \end{align*}
Thus, we can consider the action of $U_s$ and $U_\tau$ in the space $G$ as follows.
  \begin{align*}
U_s &: a\ket{\tau} + b\ket{s} \mapsto (\ket{\tau} \ket{s}) \left( \begin{matrix} -1 & 0 \\ 2/\sqrt{N} & 1 \\ \end{matrix} \right) \left( \begin{matrix} a \\ b \\ \end{matrix} \right) \\
U_\tau &: a\ket{\tau} + b\ket{s} \mapsto (\ket{\tau} \ket{s}) \left( \begin{matrix} -1 & -2/\sqrt{N} \\ 0 & 1 \\ \end{matrix} \right) \left( \begin{matrix} a \\ b \\ \end{matrix} \right) 
  \end{align*}
Similarly, we can detail the action of applying both operators sequentially in the Grover Subspace. We highlight that, while circuits proceed left-to-right, application of their matrices follows right-to-left.
  \begin{equation*}
U_s U_\tau = \left( \begin{matrix} -1 & 0 \\ 2/\sqrt{N} & 1 \\ \end{matrix} \right) \left( \begin{matrix} -1 & -2/\sqrt{N} \\ 0 & 1 \\ \end{matrix} \right) = \left( \begin{matrix} 1 & 2/\sqrt{N} \\ -2/\sqrt{N} & 1-4/N \\ \end{matrix} \right)
  \end{equation*}
We then diagonalise the above, after which we can easily reason about the behaviour of $r$ iterations. In doing so, we define $\sin(t)=1/\sqrt{N}$.
  \begin{equation*}
U_s U_\tau = M \left( \begin{matrix} \exp({2it}) & 0 \\ 0 & \exp({-2it}) \\ \end{matrix} \right) M^{-1} \text{ where } M = \left( \begin{matrix} -i & i \\ \exp({it}) & \exp({-it}) \\ \end{matrix} \right) 
  \end{equation*}
As indicated, we can now encapsulate $r$-many applications as follows.
  \begin{equation*}
{(U_s U_\tau)}^{r} = M \left( \begin{matrix} \exp({2rit}) & 0 \\ 0 & \exp({-2rit}) \\ \end{matrix} \right) M^{-1}
  \end{equation*}

Recall once more that we initialise the system to $\ket{s}$. We conclude by analysing the system following $r$ applications of our Grover operators to this initial state, projected onto the $\tau$ plane, in which the amplitude squared yields the probability of the system collapsing onto the plane under measurement.
  \begin{align*}
\left| \left(\begin{matrix} \braket{\omega}{\omega} \\ \braket{\omega}{s} \end{matrix}\right)^T (U_s U_\tau)^r \left(\begin{matrix} 0 \\ 1 \end{matrix}\right) \right|^2 = \sin^2\left((2r+1)t\right)
  \end{align*}
The probability is maximised when $(2r+1)t=\tfrac{\pi}{2}$, which we approximate to be $r=\tfrac{\pi}{4t}$. Recalling that $t=\text{arcsin}(\tfrac{1}{\sqrt{N}})$ and applying the small-angle approximation, we get that $r\approx \tfrac{\pi}{4}\sqrt{N} = \Theta(\sqrt{N})$. One can further show that the observation yields the correct answer with error $O(\tfrac{1}{N})$, but we omit such details. To conclude, we have derived that Grover's algorithm obtains the correct solution with high probability, and that the necessary number of iterations grows asymptotically as fast as $\sqrt{N}$.

\begin{theorem}[Grover Time-Complexity]
Grover's algorithm finds the unique preimage to a given output of a function, with high probability, in time $O(\sqrt{N})$ in the binary case.
\end{theorem}

  \section{A Many-Valued Generalisation} \label{section_d-ary}
In this section, we generalise Grover's algorithm to $d$-ary quantum circuits. We first address the structure of the generalised diffusion circuit, deriving a single, concise expression as in the binary case.  We continue to show that this operator is invariant under a generalised Grover Subspace. The black-box operator is left as such and simply defined with respect to this subspace: its inner-workings as a circuit are disregarded. Finally, we analyse the behaviour of the composed operators. We initially derive the time complexity in the ternary case and conjecture a generalisation of this result. However, when obtaining the final, general-case result, we follow a different approach, with justification detailed later.

Two ideas come to mind when considering how to address the generalisation: recreate the binary matrix or embrace the different natures of the various arities while preserving the abstract process followed by the algorithm. The former would likely be very complicated, as one would have to eradicate all the complex values that seep in via the $d$th roots of unity for $d>2$, and end with entries of solely $\pm1$. Thus, we follow the second approach, which we can achieve per a gate-by-gate generalisation, the circuit of which is given in Fig. 4. Intuitively, the semantics are preserved because the gate mechanisms remain invariant: a controlled phase shift conditional on the qudits being of a particular value - specified by the increment gate - all conjugated by the quantum Fourier transform.

 % --- --- ---  --- --- --- --- --- --- --- --- --- --- --- --- --- --- ---  --- --- --- %
  % --- --- --- Quantum Circuit for Grover's Algorithm in k-ary --- --- --- %
 % --- --- ---  --- --- --- --- --- --- --- --- --- --- --- --- --- --- ---  --- --- --- %
  \begin{figure}[h!]
      \begin{center}
      \begin{tikzpicture}[thick]
          % `operator' will only be used by Hadamard (H) gates here.
          % `phase' is used for controlled phase gates (dots).
          % `surround' is used for the background box.
          \tikzstyle{operator} = [draw,fill=white,minimum size=1.5em] 
          \tikzstyle{phase} = [draw,fill,shape=circle,minimum size=5pt,inner sep=0pt]
          \tikzstyle{cnot} = [draw,shape=circle,minimum size=8pt,inner sep=0pt]
          \tikzstyle{surround} = [fill=white!10,thick,draw=black,rounded corners=0mm]
          \matrix[row sep=0.4cm, column sep=0.5cm] (circuit) {
              % --- --- --- First Row --- --- --- %
              \node[operator] (H11) {$F$}; &
              \node[operator] (X12) {$X$}; &
              &
              \node[phase] (P14) {}; &
              &
              \node[operator] (X16) {{$X^{-1}$}}; & %{$X^{-1}$}; &
              \node[operator] (H17) {{$F^{-1}$}}; & \\
              %\coordinate (end1); \\

              % --- --- --- Second Row --- --- --- %
              \node[operator] (H21) {$F$}; &
              \node[operator] (X22) {$X$}; &
              &
              \node[phase] (P24) {}; &
              &
              \node[operator] (X26) {{$X^{-1}$}}; &
              \node[operator] (H27) {{$F^{-1}$}}; & \\
              %\coordinate (end2); \\

              % --- --- --- Third Row --- --- --- %
              \node[operator] (H31) {$F$}; &
              \node[operator] (X32) {$X$}; &
              \node[operator] (H33) {$F$}; &
              \node[operator] (C34) {$C_X$}; &
              \node[operator] (H35) {{$F^{-1}$}}; &
              \node[operator] (X36) {{$X^{-1}$}}; &
              \node[operator] (H37) {{$F^{-1}$}}; & \\
              %\coordinate (end3); \\
          };
         \begin{pgfonlayer}{background}
           % Draw background box.
           %\node[surround] (background) [fit = (H11) (H31) (bracket)] {};
           % Draw lines.
           \draw[thick] (H11) -- (H17)  (H21) -- (H27) (H31) -- (H37) (P14) -- (P24) (P24) -- (C34);
         \end{pgfonlayer}
      \end{tikzpicture}
      \end{center}
      \caption{A gate-by-gate $d$-ary generalisation of Grover's diffusion operator. The circuit is again presented for 3 qudits, but can be extended as highlighted in Fig. 1}
      \label{Fig. 2}
  \end{figure}
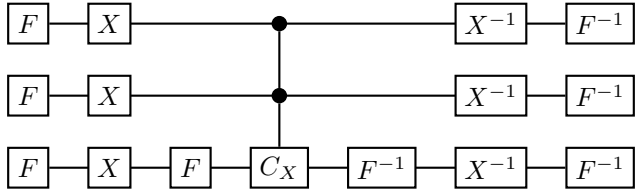

  \subsection{Structural Properties of the Generalised Diffusion Operator}
Thus, we proceed to analyse the structure of the circuit in Fig. 4. Due to the nature of the tensor product, we follow recursive proofs with prior base cases. We initially disregard the conjugation by the quantum Fourier transform, henceforth QFT, then factor it in later. We work with respect to the general dimension $d$ in all derivations. The intuition behind this first proof is that the matrix yielded by the binary operators $X \cdot H \cdot X \cdot H \cdot X$ in Lemma \ref{lemma_11} details the roots of unity descending its diagonal, which we can generalise to the following.
  \begin{equation*}
  \left(\begin{matrix}
         \omega^{d-1} & 0 & 0 & & 0 \\
         0 & \omega^{d-2} & 0 & \ldots & 0 \\
         0 & 0 & \omega^{d-3} & & 0 \\
          & \vdots & & \ddots & \vdots \\
         0 & 0 & 0 & \ldots & 1 \\
  \end{matrix}\right) = \omega^{d-1} Z^{-1}
  \end{equation*}

\begin{lemma} \label{lemma_11}
$X^{-1} F^{-1} X F X = \omega^{d-1} Z^{-1}$.
 \end{lemma}
  \begin{proof}
We can define the entries of the QFT and its inverse per $f_{i,j}=\tfrac{1}{\sqrt{d}}\omega^{ij}$ and $f_{i,j}^\dag = \tfrac{1}{\sqrt{d}}\omega^{-ij}$, respectively, indexing from $0$. Firstly, consider the internal conjugation, $F^{-1}XF$. Denoting $M = X F$, we have that $m_{i,j} = \tfrac{1}{\sqrt{d}}\omega^{(i-1)j}$ as $X_d$ specifies a row shift upwards, and denoting $N = F^{-1} M$, we have that $n_{i,j} = \sum_{k=0}^{d-1} f_{i,k}^{\dag} \cdot m_{k,j}$.
  \begin{align*}\begin{split}
n_{i,j} &= \sum_{k=0}^{d-1} f_{i,k}^{\dag} \cdot m_{k,j} \\
&= \tfrac{1}{\sqrt{d}}^2 \sum_{k=0}^{d-1} \omega^{-ik + (k-1)j} \\
&= \tfrac{1}{d} \sum_{k=0}^{d-1} \omega^{k(j-i)-j} \\
&=  \delta_{i,j} \cdot \omega^{-j}
  \end{split}\end{align*}
$\delta_{i,j}$ denotes the Kronecker delta, a function on two variables $i$ and $j$, defined to be $1$ when $i=j$ and $0$ when $i \neq j$. As we indexed from $0$, $N$ details a diagonal matrix whose descending elements are $\omega^{-0}$, $\omega^{-1}$, $\omega^{-2}$, etc. Recalling that the $d$th roots of unity form a cyclic group in which $-i \equiv d-i \ (\text{mod } d)$, the descending elements are equivalent to $\omega^0=1$, $\omega^{d-1}$, $\omega^{d-2}$, etc. The external conjugation, $X^{-1}NX$, then specifies both a row shift upwards and column shift leftwards, subsequently ordering the diagonal entries from $\omega^{d-1}$ to $1$. The resultant matrix is equivalent to $\omega^{d-1} Z^{-1}$.
  \end{proof}

We now take the circuit from Lemma \ref{lemma_11} and extend it over $n$-many qudits, generalising the Grover Diffusion Circuit Theorem from Section $3$. However, we must first provide a generalisation of the controlled NOT gate. Remembering that the common generalisation of the NOT gate is the INC gate, let us consider a controlled INC, $C_{X,d}^n$, defined such that the target qudit is incremented by $1 \ (\text{mod } d)$ only when all control qudits are value $1$. This matrix can be recursively defined as follows, in which $C_{X,d}^1=X_d$.
  \begin{equation*}
C_{X,d}^n = \left( \begin{matrix}
    I & 0 & 0 & \ldots & 0 \\
    0 & C_{X,d}^{n-1} & 0 & \ldots & 0 \\
    0 & 0 & I & \ldots & 0 \\
    \vdots & \vdots & \vdots & \ddots & \vdots \\
    0 & 0 & 0 & \ldots & I \\
\end{matrix} \right)
  \end{equation*}
As with the other gates, we henceforth omit the $d$ subscript from the controlled INC.

\begin{lemma} \label{lemma_12}
The circuit $({X^{-1}})^{\otimes n} \cdot (I^{\otimes n-1} \otimes F^{-1}) \cdot C_{X}^n \cdot (I^{\otimes n-1} \otimes F) \cdot X^{\otimes n}$ yields a block diagonal matrix, composed of $d \times d$ submatrices, with the topmost-left entry detailing the result from Lemma \ref{lemma_11} and the remaining diagonal entries specifying the identity matrix. 
\end{lemma}

This generalises the binary diffusion operator $2\ket{s}\bra{s}-I$, a priori conjugation by the Hadamard; i.e. $2\ket{0}\bra{0} - I$.

  \begin{proof}
We follow a recursive proof. Denote $\zeta_d(n)=({X_d^{-1}})^{\otimes n} \cdot (I^{\otimes n-1} \otimes F_d^{-1}) \cdot C_{X,d}^n \cdot (I^{\otimes n-1} \otimes F_d)  \cdot X_d^{\otimes n}$. Again, we drop the subscript. When $n=1$, $\zeta(1) = X^{-1} F^{-1} X F X$, is the result from Lemma \ref{lemma_11}. 

Firstly, consider the internal conjugation. The matrix $I^{\otimes n-1} \otimes F$ is a $d^n \times d^n$ block diagonal matrix, in which each submatrix along the diagonal is the $d^{n-1} \times d^{n-1}$ matrix $I^{\otimes n-2} \otimes F$; $I^{\otimes n-1} \otimes F^{-1}$ is similarly defined. Thus, working down recursively, the matrix simply comprises $F$ sub-entries along the main diagonal. It clearly follows that $I^{\otimes n-1} \otimes F$ and $I^{\otimes n-1} \otimes F^{-1}$ are the inverse of one another. Thus, the internal conjugation yields the following, in which the identity diagonal entries stem from the multiplication $(I^{\otimes n-2} \otimes F^{-1}) \cdot I \cdot (I^{\otimes n-2} \otimes F)$.
  \begin{equation*}
(I^{\otimes n-1} \otimes F^{-1}) \cdot C_{X}^{n} \cdot (I^{\otimes n-1} \otimes F) =
      \left(
      \begin{matrix}
        I & 0 & 0 & \ldots & 0 \\
        0 & (I^{\otimes n-2} \otimes F^{-1}) \cdot C_{X}^{n-1} \cdot (I^{\otimes n-2} \otimes F) & 0 & \ldots & 0 \\
        0 & 0 & I & \ldots & 0 \\
        \vdots & \vdots & \vdots & \ddots & \vdots & \\
        0 & 0 & 0 & \ldots & I \\
      \end{matrix}
      \right)
  \end{equation*}
Recalling from Lemma \ref{lemma_11} that conjugation by $X$ details a row shift upwards and column shift leftwards, conjugation per $X^{\otimes n}$ yields the same shift pattern, alongside simultaneous multiplication by $X^{\otimes n-1}$. Thus, we conclude that the top-left entry is the submatrix $({X^{-1}})^{\otimes n-1} \cdot (I^{\otimes n-2} \otimes F^{-1}) \cdot C_{X}^{n-1} \cdot (I^{\otimes n-2} \otimes F)  \cdot X^{\otimes n-1}=\zeta(n-1)$, while all other diagonal entries are specified to be the identity matrix, $({X^{-1}})^{\otimes n-1} \cdot I \cdot X^{\otimes n-1} = I_{d^{n-1}}$.
  \begin{equation*}
      \zeta(n) = 
      \left(
      \begin{matrix}
         \zeta(n-1) & 0 & 0 & & 0 \\
         0 & I_{d^{n-1}} & 0 & \ldots & 0 \\
         0 & 0 & I_{d^{n-1}} & & 0 \\
          & \vdots &  & \ddots & \vdots \\
         0 & 0 & 0 & \ldots & I_{d^{n-1}} \\
      \end{matrix}
      \right)
  \end{equation*}
One can clearly see the generalised nature of the above nature in comparison with the binary case.
  \end{proof}

We now factor in conjugation by the QFT, detailing the generalised diffusion operator in its entirety, which we henceforth denote $\Delta_d(n)$, or just $\Delta$ when the arity and argument are implicit. We follow two separate proofs: the first details the base case on a single qudit, extending Lemma \ref{lemma_11}, and the second extends over $n$-many qudits, extending Lemma \ref{lemma_12} into our first theorem.

\begin{lemma}\label{lemma_13}
$F^{-1} \zeta(1) F = F^{-1} \left( X^{-1} F^{-1} X F X \right) F = \omega^{d-1} X$.
\end{lemma}

  \begin{proof}
Recall from lemma \ref{lemma_11} that we can define the entries of the QFT and its inverse per $f_{i,j} = \omega^{ij}$, $f_{i,j}^\dag = \omega^{-ij}$, ignoring the scaling by $\tfrac{1}{\sqrt{d}}$ for now. Moreover, we can define the entries of $\zeta(1)=\omega^{d-1}Z^{-1}$ per $\zeta(1)_{i,j} = \delta_{i,j} \cdot \omega^{d-(i+1)} = \delta_{i,j} \cdot \omega^{d-(j+1)}$; both $i$ and $j$ can define the exponent, as the Kronecker delta mandates that they be equal. Considering the conjugation, let us denote $M = \zeta_d(1) \cdot F$.
  \begin{align*}\begin{split}
m_{i,j} &= \sum_{k=0}^{d-1} \zeta_d(1)_{i,k} \cdot f_{k,j} \\
&= \sum_{k=0}^{d-1} \delta_{i,k} \cdot \omega^{d-(i+1)} \cdot \omega^{kj} \\
&= \omega^{d-(i+1)} \cdot \omega^{ij}
  \end{split}\end{align*}
The last equality comes from fact that $\delta_{i,k} \neq 0 \Leftrightarrow k=i$. To complete the conjugation, denote $N = F^{-1} \cdot M$.
  \begin{align*}\begin{split}
n_{i,j} &= \sum_{k=0}^{d-1} f_{i,k}^\dag \cdot m_{k,j} \\
&= \sum_{k=0}^{d-1} \omega^{-ik} \cdot (\omega^{d-(k+1)} \cdot \omega^{kj}) \\
&= \sum_{k=0}^{d-1} \omega^{d-(k+1)+k(j-i)} \\
&= \omega^{d-1} \sum_{k=0}^{d-1} \omega^{k(j-1-i)}
  \end{split}\end{align*}
Note that, when $i \equiv j-1 \ (\text{mod } d)$, we have that $\omega^{k(j-1-i)} = \omega^{k \cdot 0} = 1$ and subsequently $n_{i,j} = \omega^{d-1} \sum_{k=0}^{d-1} 1 = d \cdot \omega^{d-1}$. On the other hand, when $i \not\equiv j-1 \ (\text{mod } d)$, the summation covers all $d$th roots of unity, yielding $n_{i,j} = \omega^{d-1} \cdot 0$.

Finally, remembering that we initially ignored the scalar coefficients of both $F$ and $F^{-1}$, all entries are then multiplied through by $\tfrac{1}{\sqrt{d}}$. The only nonzero entries are $i,j$ such that $i \equiv j-1 \ (\text{mod } d)$, which corresponds to the nonzero entries in the $X$ matrix. Thus, we conclude that $F^{-1} \zeta(1) F=\omega^{d-1} X$.
  \end{proof}

\begin{theorem} \label{theorem_1}
The generalised Grover diffusion operator $\Delta = {F^{-1}}^{\otimes n} \zeta(n) F^{\otimes n}$, in its entirety, is specified by the matrix $\Delta_d(n) = \frac{1}{d^{n-1}} \left( J_{d^{n-1}} \otimes (\omega^{d-1}X_d - I_d) \right) + I_{d^n}$.
\end{theorem}

  \begin{proof}
Recall the structure of the block matrix $\zeta(n)$ obtained in lemma \ref{lemma_12}, in which the $d^{n-1} \times d^{n-1}$ submatrix entries are defined as follows, indexed from $0$ to $d-1$.
  \begin{align*}
\zeta(n)_{i,j} &=
  \begin{cases}
\zeta(n-1) & i=j=0 \\
I_{d^{n-1}} & i=j \neq 0 \\
0_{d^{n-1}} & i \neq j
  \end{cases}
  \end{align*}
The tensor product of QFT and its inverse, $F^{\otimes n}$ and ${F^{-1}}^{\otimes n}$, per the following submatrix entries, indexed from $0$.
  \begin{align*}
f_{i,j}^{\otimes n} &= \tfrac{1}{\sqrt{d}}\cdot \omega^{ij} \cdot F^{\otimes n-1}\  ,\  f_{i,j}^{\dag^{\otimes n}} = \tfrac{1}{\sqrt{d}} \cdot \omega^{-ij} \cdot {F^{-1}}^{\otimes n-1}
  \end{align*}
Firstly, denoting $M = \zeta(n) F^{\otimes n}$, we derive the following.
  \begin{align*}
m_{i,j} = \sum_{k=0}^{d-1} \zeta(n)_{i,k} \cdot \omega^{kj} \cdot F^{\otimes n-1}
  \end{align*}
As $\zeta(n)$ is diagonal, we can disregard all but one entry in the summation.
  \begin{align*}
m_{i,j} = \zeta(n)_{i,i} \cdot \omega^{ij} \cdot F^{\otimes n-1}
  \end{align*}
To conclude the conjugation, denote $N = {F^{-1}}^{\otimes n-1} M$.
  \begin{align*}\begin{split}
n_{i,j} &= \sum_{k=0}^{d-1} \omega^{-ik} \cdot {F^{\dag}}^{\otimes n-1} \cdot m_{k,j} \\
&= \sum_{k=0}^{d-1} \omega^{-ik} \cdot {F^{\dag}}^{\otimes n-1} \cdot \zeta(n)_{k,k} \cdot \omega^{kj} \cdot F^{\otimes n-1} \\
&= \sum_{k=0}^{d-1} \omega^{k(j-i)} \cdot {F^{\dag}}^{\otimes n-1} \cdot \zeta(n)_{k,k} \cdot F^{\otimes n-1}
  \end{split}\end{align*}
Once more, recall the structure of the diagonal $\zeta(n)$, in which the top-left entry is $\zeta(n-1)$, while the remainder of the diagonal is specified by identity matrices. 
  \begin{align*}
n_{i,j} &= \omega^{0(j-i)}\cdot\Delta(n-1) + \sum_{k=1}^{d-1} \omega^{k(j-1)} \cdot I
  \end{align*}
The above summation yields two cases. Firstly, when $i=j$, the coefficient $\omega^{k(j-i)}=1$ and the entry is given to be $n_{i,j}=\Delta(n-1) + (d-1)I$. On the other hand, when $i \neq j$, the summation covers all the roots of unity with a non-zero exponent, as our arity is prime, and the entry is given per $n_{i.j}=\Delta(n-1) + \sum_{k=1}^{d-1} \omega^k I = \Delta(n-1) - I$. 

To conclude, the structure of $\Delta(n)$ can be expressed as follows, now factoring in the scaling of the QFT that we initially disregarded.
  \begin{align*}
\Delta(n)_{i,j} &=\frac{1}{d} \begin{cases} \Delta(n-1) + (d-1)I & i=j \\ \Delta(n-1) - I & i \neq j \end{cases} \\
&= \frac{1}{d^{n-1}} \left( J_{d^{n-1}} \otimes (\omega^{d-1}X - I) \right) + I
  \end{align*}
With the second derivation coming from recursively working down to the base case, in which $\Delta(1)=\omega^{d-1}X$, then factoring in the $-I_d$ terms in all entries bar the diagonal, which instead have include a global addition through $+I_d$, implying $+I_{d^n}$.
  \end{proof}

Thus, we conclude the structural analysis, deriving a concise and cooperative expression for the generalised diffusion operator $\Delta$, with which can proceed to detail behaviour.

  \subsection{The Invariant Grover Subspace}

Recall that, in the binary case, both the diffusion and black-box operators were invariant under the subspace $\{\ket{\tau},\ket{s}\}$, which we denoted the Grover Subspace. We might expect the same to be true for the generalised operators and this is indeed the case. We proceed to determine this invariant subspace, denoted the $d$\textit{-dimensional Grover Subspace} $G_d$; note that we derive this subspace solely from the generalised diffusion operator, then define the black-box operator in accordance with it. The result is presented in Theorem 2, following the intermediary work of Lemmas \ref{lemma_11} and \ref{lemma_12}. We drop the subscript from the Grover Subspace when it is implicit.

Theorem \ref{theorem_1} showed us that the matrix $\omega^{d-1}X-I$ underlies the structure of $\Delta$; let us denote this matrix $\Lambda$. As such, its properties are of importance in analysing $\Delta$, and we presently discuss them. The matrix is of rank $d-1$ and its eigenvectors constitute $d-1$ column vectors of the $d$-ary QFT, $F$. Explicitly, they are the $1$st, $3$rd, $4$th, etc columns; i.e. $f_0,f_2,\ldots,f_{d-1}$. The $2$nd column vector spans the null-space. The respective eigenvalues are $(\omega^{d-1}-1)$, $(\omega-1)$, $(\omega^2-1)$, etc up to $(\omega^{d-2}-1)$ for the $d$th column. The proof is simple and thus omitted.

In Dirac's notation, the column vectors of $F$ are denoted $\ket{f_i}$, such that $i\in[0,d-1]$. We extend $\ket{f_i}$ to define the $d^n$-dimensional vectors $\ket{s_i}$, with which we can reason about the system. Let each $\ket{s_i}$ be the block vector constituting $d^{n-1}$-many copies of the $d$-dimensional subvector $\ket{f_i}$, normalised. E.g. $\ket{s_0}=\tfrac{1}{ \sqrt{d^{n-1}} }\left(\ket{f_0}\ \ket{f_0}\ \cdots\ \ket{f_0}\right)^T$. We consider the behaviour of $\Delta$ on each $\ket{s_i}$ for $i=0,2,3,\ldots,d-1$ and $\ket{\tau}$, to construct $G$.

\begin{lemma} \label{lemma_21}
$\Delta_d\ket{s_i}=\omega^{i-1}\ket{s_i}$ for $i=0,2,3,\ldots,d-1$. 
\end{lemma}

This is to say, the behaviour of the binary diffusion operator $U_s$ on the state $\ket{s}$, $U_s\ket{s}=-\ket{s}$, generalises.

  \begin{proof}
In considering $\Delta\ket{s_i}$, we must be careful to keep our states normalised and remember that both $\Delta$ and $\ket{s_i}$ are block-structures, composed of $d \times d$ and $d \times 1$ sub-matrices and -vectors. Since $\Lambda\ket{f_i} = (\omega^{i-1}-1)\ket{f_i}$, we obtain
  \begin{align*}\begin{split}
\Delta\ket{s_i} &= \tfrac{1}{d^{n-1}} \left( J_{d^{n-1}} \otimes (\omega^{d-1}X-I) \right) \tfrac{1}{\sqrt{d^n}} \left(J_{d^{n-1},1} \otimes \sqrt{d}\ket{f_i}\right)+ \ket{s_0} \\
&= \frac{1}{d^{n-1}}\frac{\sqrt{d}}{\sqrt{d^n}} \left(\begin{matrix}d^{n-1}\Lambda\ket{f_i} \\ \vdots \\ d^{n-1}\Lambda\ket{f_i}\end{matrix}\right) + \ket{s_i} \\
&= (\omega^{i-1}-1)\ket{s_i} + \ket{s_i} \\
&= \omega^{i-1}\ket{s_i}
  \end{split}\end{align*}
Thus we have that each considered $\ket{s_i}$ is an eigenvector of the diffusion operator, with its eigenvalue determined by its composite vectors column index in the QFT: the $i$ in $\ket{f_i}$ .
  \end{proof}

The following lemma, demonstrating the behaviour of $\Delta$ on $\ket{\tau}$, is quite technical; as such, we work through an example at the end of this section. The result details that, under $\Delta$, $\ket{\tau}$ is transformed into the sum of vectors including itself and each $\ket{s_i}$ for $i=0,2,3,\cdots,d-1$ multiplied by some scaling amount, that depends on both the value of $i$ and position of the $1$-entry in $\tau$.

\begin{lemma} \label{lemma_22}
$\Delta_d\ket{\tau}=\ket{\tau}+\sum_{i \in [0,d-1]-\{1\}}f_{i,k}^\dag \tfrac{\omega^{i-1}-1}{\sqrt{N}}\ket{s_i}$, with the $k$ depending on the position of the $1$-entry within $\tau$; explicitly, the row index of the $1$ entry, modulo $d$, indexed from $0$.
\end{lemma}

  \begin{proof}
We consider $\ket{\tau}$ as a $d^n$-dimensional vector composed of $d^{n-1}$-many $d$-dimensional subvectors, as with our consideration of the $\ket{s_i}$. As $\ket{\tau}$ only contains one non-zero entry, the vector can be expressed per a collection of $d^{n-1}-1$ zero vectors and one vector $\ket{\tau_k}$ containing the one non-zero entry in its $k$th row, indexed from $0$: i.e. $\ket{\tau}=\left(0\ \cdots\ 0\ \tau_k\ 0\ \cdots\ 0\right)^T$, $\ket{\tau_k}=\left(0\ \cdots\ 0\ 1\ 0\ \cdots\ 0\right)^T$.
  \begin{align*}\begin{split}
\Delta\ket{\tau} &= \tfrac{1}{d^{n-1}} \left( J_{d^{n-1}} \otimes \Lambda \right)\ket{\tau} + \ket{\tau} \\
&= \tfrac{1}{d^{n-1}} \left(J_{d^{n-1},1} \otimes \Lambda\ket{\tau_k}\right) + \ket{\tau} \\
&= \frac{1}{d^{n-1}} \left(\begin{matrix}\Lambda\ket{\tau_k} \\ \vdots \\ \Lambda\ket{\tau_k}\end{matrix}\right) + \ket{\tau}
  \end{split}\end{align*}
The result is dependent on the given $\tau_i$ under the operator $A=\omega^{d-1}X_d - I_d$, which we can solve for all $\tau_i$ as follows. Firstly, consider that the matrix $T$ containing each $\tau_i$ as the $i$th column vector is simply the identity matrix $I_d$. Thus, $\Lambda T=\Lambda$. Consequently, we can express each $\Lambda\ket{\tau_i}$ as linear combination of the vectors $\{\ket{f_0},\ket{f_2},\ket{f_3},\cdots,\ket{f_{d-1}}\}$, as we know that $\Lambda$ has rank $d-1$ and the QFT has rank $d$. We can derive an expression for each column vector of $\Lambda$ as a linear combination of column vectors in the QFT by solving the linear system $\Lambda = F \cdot M$, for some matrix of coefficients $M$. The coefficients are then extracted by inverting the QFT and determining $M=F^{-1} \cdot \Lambda$. 
  \begin{align*}\begin{split}
F^{-1}\left(\omega^{d-1}X - I\right) &= F^{-1}\omega^{d-1}X - F^{-1} \\
\left(F^{-1}\left(\omega^{d-1}X - I\right)\right)_{i,j} &= \left(\sum_{k=0}^{d-1} f_{i,k}^\dag \cdot \delta_{k,j-1} \cdot \omega^{d-1} \right) -f_{i.j}^\dag \\
 &=  \left(\sum_{k=0}^{d-1} \omega^{-ik} \cdot \delta_{k,j-1} \cdot \omega^{d-1} \right) - \omega^{-ij} \\
&= \omega^{-i(j-1)} \cdot \omega^{d-1} - \omega^{-ij} \\
&= \omega^{-ij}\left(\omega^{i-1}-1\right)
  \end{split}\end{align*}
The $i$th column vector in $M$ details, indexing from $0$, details the coefficients in expressing $A\ket{\tau_i}$ as a linear combination of the column vectors of the QFT. Immediately, one notes that the coefficients on the $2$nd row ($i=1$) are always 0. We now substitute these into the operation $\Delta\ket{\tau}$.
  \begin{align*}\begin{split}
\Delta\ket{\tau} &= \frac{1}{d^{n-1}} \left(\begin{matrix}\Lambda\ket{\tau_k} \\ \vdots \\ \Lambda\ket{\tau_k}\end{matrix}\right) + \ket{\tau} \\
&=  \frac{1}{d^{n-1}} \left(\begin{matrix}\tfrac{1}{\sqrt{d}} \sum_i \omega^{-ik}(\omega^{i-1}-1)\ket{f_i} \\ \vdots \\ \tfrac{1}{\sqrt{d}} \sum_i \omega^{-ik}(\omega^{i-1}-1)\ket{f_i}\end{matrix}\right) + \ket{\tau} \\
&= \frac{1}{d^n} \left(\sum_i \omega^{-ik}(\omega^{i-1}-1) (\sqrt{d^n}\ket{s_i})\right) + \ket{\tau} \\
&= \frac{1}{\sqrt{d^n}} \left(\sum_{i \in [0,d-1]-{1}} \omega^{-ik}\left(\omega^{i-1}-1\right) \ket{s_i}\right) + \ket{\tau} 
  \end{split}\end{align*}
Defining $N$ to be the size of the codomain of the function being reversed by Grover's algorithm, $N=d^n$. Thus, the above details the lemma statement. 
  \end{proof}

Note the preservation of the binary operator from Section $3$, again up to a $-1$ phase shift: $\Delta_2\ket{\tau}=\ket{\tau}+\tfrac{\omega_2^1-1}{\sqrt{N}}\ket{s}=\ket{\tau}-\tfrac{2}{\sqrt{N}}\ket{s}=(-1)\cdot U_s\ket{\tau}$.

\begin{theorem} \label{theorem_2}
$\Delta_d$ is invarant under the $d$-dimensional Grover subspace $G_d = \left\langle \ket{\tau},\ket{s_0},\ket{s_2},\cdots,\ket{s_{d-1}} \right\rangle$.
\end{theorem}

  \begin{proof}
The theorem is a direct corollary from Lemmas \ref{lemma_21} and \ref{lemma_22}. Moreover, as the algorithm specifies that we initialise the system into $\ket{s_0}$, we know that this invariant subspace is maintained throughout computation.
  \end{proof}
 
Let us now consider a concrete example: ternary. The coefficients of the $\ket{s_0}$ and $\ket{s_2}$ terms are obtained from the column vectors of the specified $M = F^{-1} \Lambda \left(\tau_1\ \tau_2\ \tau_3\right)^T=F^{-1} \Lambda$.
  \begin{align*}
M &= \frac{1}{\sqrt{3}}\left(\begin{matrix} 1 & 1 & 1 \\ 1 & \omega^2 & \omega \\ 1 & \omega & \omega^2 \end{matrix}\right) \left(\begin{matrix} -1 & \omega^2 & 0 \\ 0 & -1 & \omega^2 \\ \omega^2 & 0 & -1 \end{matrix}\right) = \frac{1}{\sqrt{3}}\left(\begin{matrix} \omega^2-1 & \omega^2-1 & \omega^2-1 \\ 0 & 0 & 0 \\ \omega-1 & \omega(\omega-1) & \omega^2(\omega-1) \end{matrix}\right)
  \end{align*}
Again, note that the coefficient of $\ket{s_1}$ is always $0$, as $\Lambda$ has rank $2$ and the column vectors of the inverse QFT span $\mathbb{C}^3$. Extracting the coefficients, we detail the operation as follows.
  \begin{align*}
\Delta_3\ket{\tau} = \Delta_3 \left(0\ \ldots\ \ket{\tau_k}\ \ldots\ 0\right)^T= \frac{\omega^2-1}{\sqrt{N}} \ket{s_0} + \omega^{-2\cdot k} \frac{\omega-1}{\sqrt{N}} \ket{s_2} + \ket{\tau}
  \end{align*}
We can thus express the behaviour of $\Delta_3$ in the $3$-dimensional Grover Subspace $G_3 = \left\langle \ket{\tau},\ket{s_0},\ket{s_2} \right\rangle$ as follows. Note the phase- shift on $\ket{s_2}$, dependent upon the position of the $1$ element within $\tau$.
  \begin{align*}
\Delta_3(n) &: a\ket{\tau} + b\ket{s_2} + c\ket{s_0} \mapsto (\ket{\tau} \ket{s_2} \ket{s_0}) \left( \begin{matrix} 1 & 0 & 0 \\ \omega^{-2k}\frac{\omega-1}{\sqrt{N}} & \omega & 0 \\ \frac{\omega^2-1}{\sqrt{N}} & 0 & \omega^2 \\ \end{matrix} \right) \left( \begin{matrix} a \\ b \\ c \\ \end{matrix} \right) 
  \end{align*}

We have followed the direction of the binary proof and now need only define the extended black-box operator before considering $r$ applications of both operators, on a system starting in $\ket{s_0}$.

  \subsection{Ternary Behavioural Analysis}
Presently, we proceed with consideration solely for the ternary case, for reasons explained later. Given the structure of $\Delta$, derived in Theorem \ref{theorem_1}, and corresponding invariant Grover Subspace, derived in Theorem \ref{theorem_2}, we define the generalised black-box operator, denoted $\Upsilon$. Notice that, in the binary case, the product of the two operators $\Gamma_2=U_S U_\tau$ tends to equivalence with the identity matrix as $N \to \infty$, but the two are not inverses! With this guidance, we define $\Upsilon$ to similarly reflect a `near-inverse' of $\Delta$.
  \begin{align*}
\Delta_3 = \left(\begin{matrix}1 & 0 & 0 \\ \omega^{-2k}\frac{\omega-1}{\sqrt{N}} & \omega & 0 \\ \frac{\omega^2-1}{\sqrt{N}} & 0 & \omega^2 \end{matrix}\right) , \Upsilon_3 = \left(\begin{matrix}1 & \omega^{+2k}\frac{1-\omega^2}{\sqrt{N}} & \frac{1-\omega}{\sqrt{N}} \\ 0 & \omega^2 & 0 \\ 0 & 0 & \omega \end{matrix}\right)
  \end{align*}

For the remainder of this subsection (4.3), we shall drop the $3$ subscript from the operators, $\Delta_3$ and $\Upsilon_3$, and their composition, $\Gamma_3$. We leave the details of the black-box operator as such, disregarding circuit-specifics, proceeding onto analysis of the behaviour and performance of $\Gamma=\Delta\Upsilon$. As with the binary case, we initialise the system to state $\ket{s_0}$, then consider $r$ applications of $\Gamma$ such that $r$ maximises the probability of measuring $\ket{\tau}$. The probability is determined by the amplitude of the system projected onto the $\tau$ plane. As with the binary case, $r$ should be a periodic function on $N$ with the same rate of growth as the binary result. But first, we must know the structure of $\Gamma$, given as follows.
  \begin{align*}
\Gamma = \Delta \Upsilon = \left(\begin{matrix} 1 & \omega^{-2k}\frac{1-\omega^2}{\sqrt{N}} & \frac{1-\omega}{\sqrt{N}} \\ \omega^{2k}\frac{\omega-1}{\sqrt{N}} & 1-\frac{3}{N} & \omega^{2k}\frac{3\omega}{N}\\ \frac{\omega^2-1}{\sqrt{N}} & \omega^{-2k}\frac{3\omega^2}{N} & 1 - \frac{3}{N} \end{matrix}\right)
  \end{align*}

\begin{lemma} \label{lemma_31}
The characteristic polynomial of $\Gamma_3$ is given to be $p_3(\lambda)=(1-\lambda)\left(\lambda^2 - 2\left( 1-\tfrac{3}{N} \right)\lambda + 1\right)$, from which we extract that the eigenvalues are $1$, $e^{iT}$, and $e^{-iT}$, with $\cos(T)=1-\tfrac{3}{N}$.
\end{lemma}

This result comes from simply following the calculations necessary to obtain the determinant of $\Gamma-\lambda I$ and thus we omit the proof. However, a more interesting result is that in defining $\Upsilon_d$ to extend the intuition cemented in $\Upsilon_3$, we obtain that the general characteristic polynomial of $\Gamma_d$ is given by $p_d(\lambda)=(1-\lambda)^{d-2}\left(\lambda^2 - 2\left( 1-\tfrac{d}{N} \right)\lambda + 1\right)$ and thus the eigenvalues, for any $d$, are $1$, $e^{iT}$, and $e^{-iT}$, with $\cos(T)=1-\tfrac{d}{N}$. This precisely recreates the binary case of Section $3$, as shown below. Recall that $\sin(t)=\tfrac{1}{\sqrt{N}}$ and we derived the binary roots to be $e^{\pm2it}$.
  \begin{align*}\begin{split}
e^{\pm iT} = \exp\left(i \pm \arccos\left(1-\tfrac{2}{n}\right)\right) &= \cos\arccos\left(1-\tfrac{2}{N}\right) \pm i \sin\arccos\left(1-\tfrac{2}{N}\right) \\
&= 1-\tfrac{2}{N} \pm i \sqrt{1-\left(1-\tfrac{2}{N}\right)^2} \\
&= 1-\tfrac{2}{N} \pm i \tfrac{2}{\sqrt{N}}\sqrt{1-\tfrac{1}{n}} \\
&= \cos(2t) \pm i\sin(t)\sqrt{\cos^2(t)} \\
&= \cos(2t) \pm i\sin(2t) = e^{\pm2it}
  \end{split}\end{align*}
We refer to this generalised result later.

\begin{lemma}\label{lemma_32}
The eigenvectors of $\Gamma_3$ are the column vectors of the matrix $Q_k$, given as follows.
  \begin{align*}
Q_k  = \Sigma Q = \left(\begin{matrix} 1 & 0 & 0 \\ 0 & \omega^k & 0 \\ 0 & 0 & 1 \end{matrix}\right)\left(\begin{matrix} 0 & 1 & 1 \\  \omega^2 & \frac{1-e^{iT}}{2(\omega^2-1)\sin{t}} & \frac{1-e^{-iT}}{2(\omega^2-1)\sin{t}} \\ \omega & \frac{1-e^{iT}}{2(\omega-1)\sin{t}} & \frac{1-e^{-iT}}{2(\omega-1)\sin{t}} \end{matrix}\right) 
  \end{align*}
\end{lemma}

  \begin{proof}
  We have $\Gamma = \Sigma \Gamma' \Sigma^{-1}$, where $\Sigma$ is given above, and $\Gamma'$ is as follows.
  \[
\Gamma' = \left(\begin{matrix} 1 & \frac{1-\omega^2}{\sqrt{N}} & \frac{1-\omega}{\sqrt{N}} \\ \frac{\omega-1}{\sqrt{N}} & 1-\frac{3}{N} & \frac{3\omega}{N}\\ \frac{\omega^2-1}{\sqrt{N}} & \frac{3\omega^2}{N} & 1 - \frac{3}{N} \end{matrix}\right)
  \]
  Therefore, if the eignevectors of $\Gamma'$ are the columns of $Q$, then the eigenvectors of $\Gamma$ are the columns of $Q_k = \Sigma Q$. We thus verify that the columns of $Q$ are indeed the eigenvectors of $\Gamma'$
  
The first eigenvector is easy to extract and verify. To find the others, one must solve the homogeneous linear systems $(\Gamma'-e^{\pm iT}I)x=0$. As we follow solely the ternary case, we omit the solving of these systems and proceed to prove that the eigenvectors are as such. We follow the proof for the second eigenvector, with proof of the third being almost identical. %Consequently, consider the second column vector of $Q$ - not $Q_k$, we account for $\Sigma$ later - under the operation of $\Gamma$, in which we similarly disregard the $\omega^{\pm k}$ phase-shifts.
  \begin{align*}
\left(\begin{matrix} 1 & \frac{1-\omega^2}{\sqrt{N}} & \frac{1-\omega}{\sqrt{N}} \\ \frac{\omega-1}{\sqrt{N}} & 1-\frac{3}{N} & \frac{3\omega}{N}\\ \frac{\omega^2-1}{\sqrt{N}} & \frac{3\omega^2}{N} & 1 - \frac{3}{N} \end{matrix}\right)\left(\begin{matrix} 1 \\ \frac{1-e^{iT}}{2(\omega^2-1)\sin{t}} \\ \frac{1-e^{iT}}{2(\omega-1)\sin{t}} \end{matrix}\right) &= \left(\begin{matrix} 1 + \frac{(1-\omega^2)\sin(t)(1-e^{iT})}{2(\omega^2-1)\sin(t)} +  \frac{(1-\omega)\sin(t)(1-e^{iT})}{2(\omega-1)\sin(t)} \\ (\omega-1)\sin(t) + \frac{(1-3\sin^2(t))(1-e^{iT})}{2(\omega^2-1)\sin(t)} + \frac{3\omega\sin^2(t)(1-e^{iT})}{2(\omega-1)\sin(t)} \\ (\omega^2-1)\sin(t) + \frac{3\omega^2\sin^2(t)(1-e^{iT})}{2(\omega^2-1)\sin(t)}  + \frac{(1-3\sin^2(t))(1-e^{iT})}{2(\omega-1)\sin(t)}\end{matrix}\right) &
  \end{align*}

Verifying that the first entry details multiplication by $\lambda=e^{iT}$ is easy, but the second and third entries require more work. We prove the former, with the latter following the same derivation.
  \begin{align*}\begin{split}
&(\omega-1)\sin(t) + \frac{(1-3\sin^2(t))(1-e^{iT})}{2(\omega^2-1)\sin(t)} + \frac{3\omega\sin^2(t)(1-e^{iT})}{2(\omega-1)\sin(t)} \\
=& \frac{2(\omega^2-1)(\omega-1)\sin^2(t)}{2(\omega^2-1)\sin(t)} + \frac{(1-3\sin^2(t))(1-e^{iT})}{2(\omega^2-1)\sin(t)} - \frac{3\sin^2(t)(1-e^{iT})}{2(\omega^2-1)\sin(t)} \\
=& \frac{6\sin^2(t) + 1 - e^{iT} - 3\sin^2(t) + 3e^{iT}\sin^2(t) - 3\sin^2(t) + 3e^{iT}\sin^2(t)}{2(\omega^2-1)\sin(t)} \\
=& \frac{6e^{iT}\sin^2(t) + (1-e^{iT})}{2(\omega^2-1)\sin(t)}
  \end{split}\end{align*}

We conclude equivalence by demonstrating that $6e^{iT}\sin^2(t) + (1-e^{iT}) = e^{iT} \cdot (1-e^{iT})$, remembering that we defined $\cos(T)=1-\tfrac{3}{N}$ and $\sin(t)=\tfrac{1}{\sqrt{N}}$. 
  \begin{align*}\begin{split}
e^{iT} &= \cos\arccos\left(1-\tfrac{3}{N}\right) + i \sin\arccos\left(1-\tfrac{3}{N}\right) \\
&= 1-\tfrac{3}{N} + i\sqrt{1-\left(1-\tfrac{3}{N}\right)^2} \\ 
&= 1-\tfrac{3}{N} - \sqrt{\tfrac{9}{N^2} - \tfrac{6}{N}}
  \end{split}\end{align*}
  \begin{align*}\begin{split}
(e^{iT})^2 &= \cos(2T) + i\sin(2T) = 2\cos^2(T)-1 + 2i\sin(T)\cos(T) \\
&= 2\left(1-\tfrac{3}{N}\right)^2-1 + 2i\sqrt{1-\left(1-\tfrac{3}{N}\right)^2}\left(1-\tfrac{3}{N}\right) \\
&= 1 - \tfrac{12}{N} + \tfrac{18}{N^2} - 2\sqrt{\tfrac{9}{N^2} - \tfrac{6}{N}}\left(1-\tfrac{3}{N}\right) 
  \end{split}\end{align*}
Denote $\xi=\sqrt{\tfrac{9}{N^2} - \tfrac{6}{N}}$.
  \begin{align*}\begin{split}
e^{iT} \cdot \left(1-e^{iT}\right) = e^{iT} - (e^{iT})^2 &= 1 - \tfrac{3}{N} - \xi - 1 + \tfrac{12}{N} - \tfrac{18}{N^2} + 2\xi\left(1-\tfrac{3}{N}\right) \\
&= \tfrac{9}{N} - \tfrac{18}{N^2} + \xi\left(1-\tfrac{6}{N}\right) \\
6e^{iT}\sin^2(t) + (1-e^{iT}) &= \tfrac{6}{N}\left(1-\tfrac{3}{N}-\xi\right) + 1 - \left(1-\tfrac{3}{N}-\xi\right) \\
&= \tfrac{9}{N} - \tfrac{18}{N^2} + \xi\left(1-\tfrac{6}{N}\right) \\
&= e^{iT} \cdot (1-e^{iT})
  \end{split}\end{align*}
%We now account for $\Sigma$. Factoring in the phase shifts, we see that with the summation of the first and third rows of $\Gamma$, the $\omega^k$ and $\omega^{-k}$ terms cancel out, and that when the second row is applied, all terms in the summation have the coefficient $\omega^k$, which we thus factor into the eigenvectors as $\Sigma$, yielding $Q_k$.
  \end{proof}

With the eigenvectors determined, we diagonalise $\Gamma = Q_k^{-1} \cdot D \cdot Q_k$, from which reasoning about the behaviour of $r$ iterations is simplified to $\Gamma^r = Q_k^{-1} \cdot D^r \cdot Q_k$, with the exponent considered solely apropos the diagonal $D$. $Q_k^{-1}$ is as follows; we omit proof.
  \begin{align*}
Q_k^{-1} = Q^{-1}\Sigma^{-1} = \left(\begin{matrix} 0 & \tfrac{1}{2}\omega & \tfrac{1}{2}\omega^2 \\ -\frac{1-e^{-iT}}{e^{-iT}-e^{iT}} & \frac{(\omega^2-1)\sin(t)}{e^{-iT}-e^{iT}} & \frac{(\omega-1)\sin(t)}{e^{-iT}-e^{iT}} \\ \frac{1-e^{iT}}{e^{-iT}-e^{iT}} & \frac{(1-\omega^2)\sin(t)}{e^{-iT}-e^{iT}} & \frac{(1-\omega)\sin(t)}{e^{-iT}-e^{iT}} \end{matrix}\right) \left(\begin{matrix} 1 & 0 & 0 \\ 0 & \omega^{-k} & 0 \\ 0 & 0 & 1 \end{matrix}\right)^{-1}
  \end{align*}

Nearing our conclusion, we proceed to consider the behaviour of $r$ Grover iterations on the system initially set to state $\ket{s_0}$, measured with respect to $\tau$. In determining the probability of the measurement yielding $\tau$, we project the system onto the subspace spanned by $\tau$ and square the amplitude. Overall, this is encapsulated through the following.
  \begin{align*}
\left|\bramatket{\tau}{\Gamma^r}{s_0}\right|^2 = \left| \left(\begin{matrix} \braket{\tau}{\tau} \\ \braket{\tau}{s_0} \\ \braket{\tau}{s_2} \end{matrix}\right)^T Q \left(\begin{matrix} 1 & 0 & 0 \\ 0 & e^{irT} & 0 \\ 0 & 0 & e^{-irT} \end{matrix}\right) Q^{-1} \left(\begin{matrix} 0 \\ 0 \\ 1 \end{matrix}\right) \right|^2
  \end{align*}
The probability is thus as follows.
  \begin{align*}\begin{split}
&\left| \left(\begin{matrix} 1 \\ \frac{1}{\sqrt{N}} \\ \frac{\omega^{i-1}}{\sqrt{N}} \end{matrix}\right)^T \left(\begin{matrix} \frac{(1-\omega)\sigma\sin(rT)}{\sin(T)} \\ \omega\frac{1}{2}(1-\frac{\sin((r+1)T)-\sin(rT)}{\sin(T)}) \\ \frac{1}{2}(1+\frac{\sin((r+1)T)-\sin(rT)}{\sin(T)}) \end{matrix}\right) \right|^2  \\
= & \left| (1-\omega)\sin(t)\frac{\sin(rT)}{\sin(T)} + \frac{1}{2}(1-\omega)\sin(t)\frac{\sin((r+1)T)-\sin(rT)}{\sin(T)} + \frac{1}{2}(1+\omega)\sin(t) \right|^2
  \end{split}\end{align*}

We defined $\sin(t)=\tfrac{1}{\sqrt{N}}$ and $\cos(T)=1-\tfrac{3}{N}$; they are functions depending on the value of $N$ and thus invariant under $r$. Consequently, they can be treated as constants maximising the probability through varying $r$. Collecting all constants in the above, we simplify its expression to the following.
  \begin{align*}
\left| a \sin(rT) + b\left(\sin((r+1)T)-\sin(rT)\right) + c \right|^2
  \end{align*}
Assuming that $r$ does vary with $N$ - and not that it is simply a constant - we infer that for large $N$, $\sin\left((r+1)T\right) \approx \sin(rT)$, cancelling the second term. In maximising what remains, $\left|a\sin(rT) + c\right|^2$, it is clear to see that is achieved when $rT = \tfrac{\pi}{2}$, or equivalently $r=\tfrac{\pi}{2\text{arccos}(1-3/N)}$. We now address the rate of growth of this function.

\begin{theorem} \label{theorem_3}
In the ternary case of the considered generalisation of Grover's algorithm, $r=O(\sqrt{N})$.
\end{theorem}

  \begin{proof}
By definition, we have that $r=O(\sqrt{N})$ if and only if there exist positive constants $C,k$ such that $r \leq C \cdot \sqrt{N}$ for all $N \geq k$; or equivalently, that $\tfrac{r}{\sqrt{N}}$ is bounded for all $N$ greater than some positive constant. The minimum value of $N$ that we consider is when given a single qutrit, $N=3$, at which point the function $\tfrac{r}{\sqrt{N}}$ evaluates to $\tfrac{\pi}{2\sqrt{3}\arccos(1-3/3)} = \tfrac{1}{\sqrt{3}}$. The derivative of $\tfrac{r}{\sqrt{N}}$ with respect to $N$ is as follows.
  \begin{align*}
\frac{\text{d}}{\text{d}N} \frac{\pi}{2\sqrt{N}\arccos(1-\tfrac{3}{N})} = (-1) \cdot \frac{\pi}{2N\left(\arccos(1-\tfrac{3}{N})\right)^2} \cdot \left( \frac{\arccos(1-\tfrac{3}{N})}{2\sqrt{N}} + \frac{3\sqrt{N}}{\sqrt{\tfrac{6}{N} - \tfrac{9}{N^2}}}\right)
  \end{align*}
As $N\to\infty$, $\arccos\left(1-\tfrac{3}{N}\right) \to 0$, but always remains positive. Thus, all the terms in the above derivative - to the right of the $\times (-1)$ - are positive; subsequent multiplication through $-1$ makes the derivative is invariantly negative, and thus the function never increases from the initial boundary condition, $\tfrac{1}{\sqrt{3}}$. Moreover, the function itself is always positive - tending from $\tfrac{1}{\sqrt{3}}$ to $0$. It can thus be bounded and so we have that $r=O(\sqrt{N})$.
  \end{proof}

We conjecture how this result could generalise, before proceeding, for reasons that will become apparent. We have found that the time-complexity in ternary is given to be $\tfrac{\pi}{2T}$, in which $\cos(T)=1-\tfrac{3}{N}$ is derived from the characteristic polynomial of $\Gamma$. When we derived this, we highlighted that the general case characteristic polynomial would lead us to instead extract eigenvalues $e^{\pm iT}$ in which $\cos(T)=1-\tfrac{d}{N}$, a natural generalisation. Thus, we could conjecture that the general-case time complexity would be $\tfrac{\pi}{2T}$, in which $T$ is now defined in the general case. This would preserve the $O(\sqrt{N})$ nature of the complexity, and so generalise the above theorem.

Fig. (5) visualises the growth of this conjectured, general-case complexity. Fig. (5a) and (5b) demonstrate that the conjectured result, in the binary case, tends to equality with the result derived in Section 3. Fig. 5(c) visualises the conjectured complexity in the ternary and quinary cases, again in comparison to the binary result derived in Section 3, from which one can see that the general-case result is indeed $O(\sqrt{N})$.

\begin{figure}[h!]
  \centering
  \begin{subfigure}[t]{0.3\linewidth}
    \includegraphics[width=\linewidth]{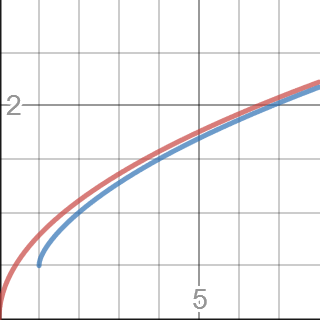}
    \centering
    \caption{Graph plotting $y=\tfrac{\pi}{4}\sqrt{x}$ (red) and $y=\tfrac{\pi}{2\arccos(1-2/x)}$ (blue), over an initial, small range.}
  \end{subfigure}
  \begin{subfigure}[t]{0.3\linewidth}
    \includegraphics[width=\linewidth]{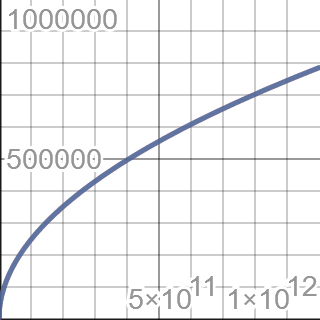}
    \centering
    \caption{Graph plotting $y=\tfrac{\pi}{4}\sqrt{x}$ (red) and $y=\tfrac{\pi}{2\arccos(1-2/x)}$ (blue) tending to equality. They overlap almost precisely.}
  \end{subfigure}
  \begin{subfigure}[t]{0.3\linewidth}
    \includegraphics[width=\linewidth]{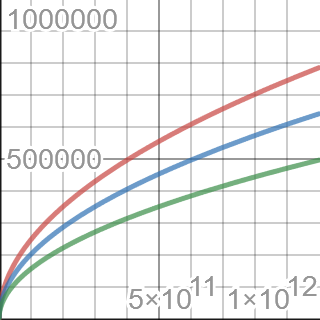}
    \centering
    \caption{Graph plotting $y=\tfrac{\pi}{4}\sqrt{x}$ (red), $y=\tfrac{\pi}{2\arccos(1-3/x)}$ (blue), and $y=\tfrac{\pi}{2\arccos(1-5/x)}$ (green).}
  \end{subfigure}
  \caption{Three graphs demonstrating the growth of function $y=\tfrac{\pi}{2\arccos(1-d/x)}$ for various $d$.}
  \label{fig:plots}
\end{figure}

  \subsection{General Behavioural Analysis}

  \subsubsection{Foreword}
Before proceeding, we discuss the necessity for a novel tactic. Our ternary proof followed the same structure as the binary and it seems natural to assume that one could further generalise this method; in which case, the desired number of iterations would be $r=\tfrac{\pi}{2T}$, $\cos(T)=1-\tfrac{d}{N}$. However, our progress in such adaptations is halted. Recall that we stated the characteristic polynomial of $\Gamma_d$ to be $p_\lambda(\Gamma)=\left(\lambda-1\right)^{d-2}\left(\lambda^2-2\left(1-\tfrac{d}{N}\right)+1\right)$, with eigenvalues $1$ and $e^{\pm iT}$, $\cos(T)=1-\tfrac{d}{N}$. In the ternary case, deriving the $e^{\pm iT}$-eigenvectors, necessary for diagonalisation, requires proof of the following equality.
  \begin{align*}
6e^{\pm iT}\sin^2(t)+(1-e^{\pm iT})=e^{iT}(1-e^{\pm iT})
  \end{align*}
In the general case, the $6$ coefficient is replaced with $(d-1)(\omega^k-1)(\omega^{-k}-1)=2(d-1)\left(1-\cos\left(\tfrac{2\pi k}{d}\right)\right)$. Hence, when $d=2$ or $3$, the solution is easy to extract regardless of the value of $k$, but a general solution cannot be determined and we cannot readily proceed.

  \subsubsection{Synthetic Analysis}
Thus, changing tactics, we now prove that the number of iterations does remain $O(\sqrt{N})$, for any $d$. Instead of diagonalsing the matrix $\Gamma$, we simply approximate it per $\Gamma^r = \exp( r/\sqrt{N} \Phi)$, for a matrix $\Phi$ that we discuss overleaf. In particular, this allows us to compute the probability of measuring $\ket{\tau}$ after $r = O(\sqrt{N})$ iterations.

\begin{theorem} \label{theorem_d-ary_Grover}
Grover's algorithm finds the unique preimage to a given output of a function, with high probability, in time $O(\sqrt{N})$ in the $d$-ary case for all $d$.
\end{theorem}

The rest of the section is devoted to the proof of Theorem \ref{theorem_d-ary_Grover}. We first consider the aformentioned approximation of $\Gamma$, for which we must recall $\Delta$ from Lemma \ref{lemma_21} and define $\Upsilon$ in the general case.
  \begin{align*}
\Delta = \left(\begin{matrix} 1 & 0 & 0 & & 0 \\
\omega^{-2k}\tfrac{\omega-1}{\sqrt{N}} & \omega & 0 & & 0 \\
\omega^{-3k}\tfrac{\omega^2-1}{\sqrt{N}} & 0 & \omega^2 & & 0 \\
\vdots & \vdots & & \ddots & \vdots \\
\tfrac{\omega^{-1}-1}{\sqrt{N}} & 0 & 0 & \ldots & \omega^{d-1} \end{matrix}\right),\ 
\Upsilon = \left(\begin{matrix} 1 & \omega^{2k}\tfrac{1-\omega^{-1}}{\sqrt{N}} & \omega^{3k}\tfrac{1-\omega^{-2}}{\sqrt{N}} & & \tfrac{1-\omega}{\sqrt{N}} \\
0 & \omega^{d-1} & 0 & \ldots & 0 \\
0 & 0 & \omega^{d-2} & & 0 \\
\vdots & \vdots & & \ddots & \vdots \\
0 & 0 & 0 & \ldots & \omega \end{matrix}\right)
  \end{align*}
Thus we can determine the matrix $\Gamma=\Delta\Upsilon$, in which $a_i$ and $b_{i,j}$ denote constants with respect to $N$.
  \begin{align*}
\Gamma=\left(\begin{matrix} 1 & \omega^{2k}\tfrac{1-\omega^{-1}}{\sqrt{N}} & \omega^{3k}\tfrac{1-\omega^{-2}}{\sqrt{N}} & & \tfrac{1-\omega}{\sqrt{N}} \\
\omega^{-2k}\tfrac{\omega-1}{\sqrt{N}} & 1+ \tfrac{1}{N}a_{1} & \tfrac{1}{N}b_{1,2} & \ldots & \tfrac{1}{N}b_{1,d-1} \\
\omega^{-3k}\tfrac{\omega^2-1}{\sqrt{N}} & \tfrac{1}{N}b_{2,1} & 1+\tfrac{1}{N}a_{2} & & \tfrac{1}{N}b_{2,d-1} \\
\vdots & \vdots & & \ddots & \vdots \\
\tfrac{\omega^{-1}-1}{\sqrt{N}} & \tfrac{1}{N}b_{d-1,1} & \tfrac{1}{N}b_{d-1,2} & \ldots & 1 + \tfrac{1}{N}a_{d-1} \end{matrix}\right)
  \end{align*}
We see that $\Gamma$ can be expressed as the sum $\Gamma = I + \Omega = I + \frac{1}{\sqrt{N}} \Phi + \frac{1}{N} \Psi$, in which $\Psi$ and the previously noted $\Phi$ are as follows.
  \begin{align*}
\Phi&=\left(\begin{matrix} 0 & \omega^{2k} (1-\omega^{-1}) & \omega^{3k}(1-\omega^{-2}) & & 1-\omega \\
\omega^{-2k}(\omega-1) & 0 & 0 & \ldots & 0 \\
\omega^{-3k}(\omega^2-1) & 0 & 0 & & 0 \\
\vdots & \vdots & & \ddots & \vdots \\
\omega^{-1}-1 & 0 & 0 & \ldots & 0 \end{matrix}\right)\\
\Psi&=\left(\begin{matrix} 0 & 0 & 0 & & 0 \\
0 & a_{1} & b_{1,2} & \ldots & b_{1,d-1} \\
0 & b_{2,1} & a_{2} & & b_{2,d-1} \\
\vdots & \vdots & & \ddots & \vdots \\
0 & b_{d-1,1} & b_{d-1,2} & \ldots & a_{d-1} \end{matrix}\right)
  \end{align*}

%All the elements of the former are of order $O(\tfrac{1}{\sqrt{N}})$ and all elements of the latter are of order $O(\tfrac{1}{N})$.

We consider $r$ applications of $\Gamma$ upon the uniformly-distributed starting state $\ket{s}$. For $r$ large, say $r = \rho \sqrt{N}$ for some $\rho > 0$, we can approximate $\Gamma^r$ as follows.
\[
    \Gamma^r = \left( I + \frac{1}{r}(r \Omega) \right)^r \approx \exp(r\Omega) = \exp \left( \rho \Phi + \frac{\rho}{\sqrt{N}} \Psi \right) \approx \exp(\rho \Phi).
\]

% In the ternary case, we chose $r$ to maximise the probability of measuring $\ket{\tau}$, but encounter problems when trying to directly generalise this work. Thus, we follow a different approach, in which we work with a specific $r$ from the beginning. Recall our conjecture that $r=\tfrac{\pi}{2\arccos(1-d/N)}$ in the general case. Since $N$ is large, we can approximate $\arccos\left(1-\tfrac{d}{N}\right) \approx \sqrt{2d/N}$. Thus, $r=\tfrac{\pi}{2\sqrt{2d}}\sqrt{N}$. Remembering that $\Gamma=I+\Omega$, we make use of the following argument.
%  \begin{align*}
%\left(I + \frac{1}{a}X\right)^a &= {a \choose 0} + {a \choose 1} \frac{1}{a} X + {a \choose 2} \frac{1}{a^2} X^2 + {a \choose 3} \frac{1}{a^3} X^3 + \ldots \\
%&= 1 + X + \frac{a!}{(a-2)!2!a^2}X^2 + \frac{a!}{(a-3)!3!a^3}X^3 + \ldots
%  \end{align*}
%When $a$ is very large, $\tfrac{a!}{a^i(a-i)!} \approx 1$ and the above summation approximates $\exp(X)$. Thus, we consider $\left(I+\tfrac{1}{r\sqrt{N}}(r\Omega)\right)^{r\sqrt{N}}$. By further analysis $\Gamma$, we see that we can express $\Omega$ alternatively as $\Omega=\Phi+\Psi$. $\Psi=O\left(\tfrac{1}{\sqrt{N}}\right)$ and as $N$ is very large we can disregard the $\Psi$ matrix. By the continuity of the exponential function, we can simplify the problem as follows.
%  \begin{align*}
%\Gamma^{r\sqrt{N}}=\left(I+\frac{1}{r\sqrt{N}}\Omega\right)^{r\sqrt{N}} \approx \exp(\Omega) = \exp(\Phi+\Psi) = \exp(\Phi) 
%  \end{align*}

  \begin{lemma} \label{lemma_phi}
For all $k\geq 0$ we have $\Phi^{2k}=\left(\begin{matrix} (-2d)^k & 0 \\ 0 & M^k \end{matrix}\right)$ for some $(d-1)\times(d-1)$ matrix $M$ and $\Phi^{2k+1}=(-2d)^k\Phi$.
  \end{lemma}

  \begin{proof}
For even powers of $\Phi$, the claim trivially holds for $\Phi^0$. We now verify it for $\Phi^2$. Let $Y=\Phi^2$, then clearly $y_{0,j}=y_{j,0}=0$ $\forall 1 \leq j \leq d-1$.
  \begin{align*}
y_{0,0} &= \sum_{j=1}^{d-1} \omega^{-(j+1)\sigma} (\omega^j - 1) \omega^{(j+1) \sigma} (1 - \omega^{-j}) \\
&= \sum_{j=1}^{d-1} (\omega^j + \omega^{-j} - 2) \\
&= 2 \left( \sum_{j=1}^{d-1} \omega^j - (d-1)\right) \\
&= -2d
  \end{align*}
Then, $\forall k\geq 1$,
  \begin{align*}
\Phi^{2k}=(\Phi^2)^k = \left(\begin{matrix} -2d & 0 \\ 0 & M \end{matrix}\right)^k = \left(\begin{matrix}(-2d)^k & 0 \\ 0 & M^k \end{matrix}\right)
  \end{align*}
The proof for odd powers of $\Phi$ is by induction on $k$. It is trivial for $k=0$ and easy to verify for $k=1$ i.e. $\Phi^3$. Now, suppose the claim holds for $k-1$.
  \begin{align*}
\Phi^{2k+1}=\Phi^{2k-1}\Phi^2=\left((-2d)^{k-2}\Phi\right)\Phi^2=(-2d)^{k-2}\Phi^3=(-2d)^{k-2}\left((-2d)^2\Phi\right)=(-2d)^k\Phi
  \end{align*}
  \end{proof}

  \begin{lemma} \label{lemma_xi}
Letting $\Xi:=\exp( \rho \Phi)$, we have $\Xi_{0,d-1}=\tfrac{1}{\sqrt{2d}}(1-\omega)\sin( \rho \sqrt{2d})$.
  \end{lemma}

  \begin{proof}
  By Lemma \ref{lemma_phi}
  \begin{align*}
\Xi_{0,d-1} &= \sum_{n=0}^\infty \rho^n \frac{\Phi_{0,d-1}^n}{n!} \\
&= \sum_{k=0}^\infty  \rho^{2k+1} \frac{(-2d)^k (1- \omega) }{(2k+1)!} \\
&= \frac{1}{\sqrt{2d}}(1-\omega) \sum_{k=0}^\infty \frac{(-1)^k (\rho \sqrt{2d})^{2k+1}}{(2k+1)!} \\
&= \frac{1}{\sqrt{2d}}(1-\omega)\sin( \rho \sqrt{2d})
  \end{align*}
  \end{proof}

Thus, according to Lemma \ref{lemma_xi}, the probability of measuring  $\ket{\tau}$, denoted as $P_d(\rho)$, is as follows, omitting any terms in $O\left(\tfrac{1}{\sqrt{N}}\right)$.
  \begin{align*}
P_d(\rho) = \left|\bramatket{\tau}{\Xi}{s}\right|^2 \approx \left| \Xi_{0,d-1} \right|^2 = \frac{1}{2d}\left|1-\omega\right|^2 \sin^2(\rho \sqrt{2d}) = \left\{ \frac{1}{d} - \frac{1}{d}\cos\left(\frac{2\pi}{d}\right) \right\} \sin^2(\rho \sqrt{2d})
  \end{align*}
This is maximised for $\hat{\rho} = \frac{\pi}{2\sqrt{2d}}$, thus yielding a maximal probability of $P_d(\hat{\rho}) = \frac{1}{d} - \frac{1}{d}\cos\left(\frac{2\pi}{d}\right)$.

How can we use such an algorithm? We see that $P_2(\hat{\rho}) \approx 1$, but for $d>2$, $P_d(\hat{\rho}) <1$. Thus, we must run the algorithm mutliple times. It is clear that, for each run, the probability of measuring $\ket{\tau'}$ for any $\tau' \neq \tau$ is $O(\tfrac{1}{\sqrt{N}})$. Recalling that $N=d^n$, this means that the probabilities of measuring any other $\tau'$ are far smaller than that of measuring $\tau$. Therefore, one can run Grover's algorithm multiple times until the same measurement is obtained twice; that repeated measurement would almost surely be $\ket{\tau}$. The expected number of runs is $\tfrac{2}{P_d(\rho)}$, which yields a total of calls to the oracle given as follows.
\[
    E(\rho) = \frac{2}{P_d(\rho)} \cdot \rho \sqrt{N} = \frac{4d}{|1 - \omega|^2} \cdot \frac{\rho}{\sin^2(\rho \sqrt{2d})} \cdot \sqrt{N}
\]
A simple exercise shows that this is minimised for $\rho^* = \frac{x}{\sqrt{2d}}$, where $x$ satisfies $x = \frac{1}{2} \tan x$. We obtain the slightly smaller expected number of calls
\[
  E(\rho^*) = \frac{\sqrt{2d}}{1 - \cos(2\pi/d)} \cdot 1.380 \cdot \sqrt{N}.  
\]
For instance, for $d=3$, we obtain $E(\rho^*) = 2.253 \sqrt{N}$, while $E(\hat{\rho}) = 2.565 \sqrt{N}$.

  \section{Conclusion} \label{section_conclusion}
We conclude having considered Grover's algorithm under the many-valued generalisation of the quantum circuit model and, subsequently, derived three main results: Theorem \ref{theorem_1} shows that our generalisation reflects the structure of the binary case; Theorem \ref{theorem_2} demonstrates that the algorithm remains invariant under a $d$-dimensional subspace, in which $d$ is the arity of the information and logic; Theorems \ref{theorem_3} and \ref{theorem_d-ary_Grover} highlight that the general case retains $O(\sqrt{N})$ time-complexity. Further, we show that higher arities require more runs of the algorithm, due to the diminishing (maximal) amplitude of the target state as $d\to\infty$, but that this expectation still remains $O(\sqrt{N})$.

Overall, this paper presents a detailed consideration of quantum algorithms under this framework, with strong results. Quantum information is more unstable than its classical counterpart and one contemporary objective in our engineering of such data is maintaining coherent states for extended periods of time. Thus, the potential for faster computation in higher-arity systems would allow quantum computers to perform more work in the same, limited time. Moreover, research into whether the circuit-generalisation considered in this paper always preserves semantics and complexity could provide insight into the nature of quantum gates and circuits, and how applications of the former transcend the arity in which they are embedded. A final question, centred  in complexity theory, could be to push this generalised notion even further and determine whether quantum classes like \texttt{BQP} are \textit{robust}, in that the class is invariant under the alphabet size of the quantum Turing machine, and whether this can be alternatively expressed with arity in the quantum circuit model.

\end{document}